\documentclass[11pt, a4paper]{article}

\usepackage[english]{babel} 
\usepackage[utf8]{inputenc}  

\usepackage{enumitem}

\usepackage{comment}
\usepackage{amssymb,amsfonts,amsmath}

\usepackage{scalerel}

\usepackage{amsthm}

\usepackage{float}
\usepackage{pgfplots}
\usepackage{pgfplotstable}
\usepgfplotslibrary{external} 
\usepackage{pdfpages}

\usepackage{subcaption}
\pgfplotsset{compat=1.16}

\newtheorem{theorem}{Theorem}[section] 
\newtheorem{lemma}[theorem]{Lemma}     
\newtheorem{corollary}[theorem]{Corollary}
\newtheorem{proposition}[theorem]{Proposition}

\newtheorem{conjecture}[theorem]{Conjecture}

\theoremstyle{definition}
\newtheorem{remark}{Remark}
\newtheorem{definition}{Definition}
\newtheorem{example}{Example}
\newtheorem{problem}{Problem}
\newtheorem{claim}{Claim}

\newtheorem{hypothesis}{Hypothesis}

\usepackage[hyphens]{url}
\usepackage{hyperref}

\usepackage{graphicx}
\usepackage{tikz}
\usetikzlibrary{arrows}
\usepackage{pgfplots,pgfplotstable}
\usepackage{multirow}
\usepackage{multicol}
\usepackage{algorithmic}
\usepackage{authblk}

\renewcommand{\rho}{\varrho}

\usepackage{fullpage}
\usepackage[onelanguage,ruled,vlined,linesnumbered,norelsize]{algorithm2e}

\newcommand{\Esp}{{\mathbb E}}
\newcommand{\F}{{\mathbb F}}

\newcommand{\Q}{{\mathbb Q}}
\newcommand{\Z}{{\mathbb Z}}

\newcommand{\EE}{{\mathbb{E}}}

\newcommand{\N}{{\mathbb N}}

\newcommand{\R}{{\mathbb R}}
\newcommand{\C}{{\mathbb C}}
\newcommand{\OO}{{\mathcal O}}
\newcommand{\cA}{{\mathcal A}}
\newcommand{\cH}{{\mathcal H}}

\newcommand{\cP}{{\mathcal P}}

\newcommand{\norm}[1]{{\parallel#1\parallel}}

\DeclareMathOperator{\Prob}{\mathrm{Prob}}

\DeclareMathOperator{\Li}{\mathrm{Li}}

\DeclareMathOperator{\val}{val}

\title{ECM and the Elliott--Halberstam conjecture for quadratic fields}
\author[1]{Razvan Barbulescu and Florent Jouve }
\affil[1]{ Univ. Bordeaux, CNRS, Bordeaux INP, IMB, UMR 5251,  F-33400, Talence, France\\  \url{razvan.barbulescu@u-bordeaux.fr} and \url{florent.jouve@u-bordeaux.fr}}
\date{January 2023}

\begin{document}
\maketitle

\begin{abstract} The complexity of the elliptic curve method of factorization (ECM) is proven under a strong conjectural form of existence of friable numbers in short intervals. In the present work we use friability to tackle a different version of ECM which is much more studied and implemented, especially because it enables the use of ECM-friendly curves. In the case of curves with complex multiplication (CM) we replace heuristic arguments by rigorous results conditional on the Elliott--Halberstam (EH) conjecture. The proven results mirror recent work concerning the count of primes $p$ such that $p-1$ is friable. In the case of non CM curves, we explore consequences of a hypothetical statement that can be seen as an elliptic curve analogue of EH.     
\end{abstract}

\section{Introduction}

Let $E/\Q$ be an elliptic curve that has good reduction precisely at every prime number not dividing an integer $\Delta_E$. The main object of study of this paper is the prime counting function $\psi_E(x,y)$ which is defined as the cardinality of
\begin{equation} \label{eq:psiE}
\Psi_E(x,y)=\{p\leq x\colon p\text{ prime, }p\nmid \Delta_E, |E(\F_p)|\text{ is $y$-friable} \}.
\end{equation}
Here we make the usual slight abuse of notation and write $E(\F_p)$ for the set of $\F_p$-points on the reduction of $E$ modulo $p$ and we also recall that an integer is \emph{$y$-friable} (or \emph{$y$-smooth}) if all its prime factors are less than~$y$. For any integer $n$ we shall denote by $P^-(n)$ (resp $P^+(n)$) the smallest (resp. the largest) prime factor of $n$. By convention $P^-(1)=\infty$. The notation $\psi_E(x,y)$ is reminiscent of $\psi(x,y)$ which denotes the cardinality of 
\begin{align}\label{eq:psi-fct}
\Psi(x,y)=\{n\leq x\colon P^+(n) < y\}.
\end{align}
Our main motivation for studying $\Psi_E(x,y)$ comes from cryptography and more precisely from the method of factorization ECM. First recall the principle and purposes of ECM~\cite{Len87}. Let $P\in E(\Q)$ be a rational point with homogeneous coordinates $P=(x_P:y_P:z_P)\in E(\Q)$, relatively to a fixed projective embedding of $E$. Without loss of generality we can assume $x_P,y_P,z_P\in \Z$. Let $N$ be a given positive integer for which one would like to find the prime factorization; set two parameters $u=u(N)$ and $v=v(N)$ in $(0,1)$ and define $B=N^{1/u}$, $C=B^{1/v}$. Note that if $\gcd(x_P,y_P,z_P,N)=1$ and $\gcd(N,\Delta_E)=1$ then $E$ has good reduction modulo any unknown prime factor $p$ of $N$ and $\bar{P}:=(x_P:y_P:z_P) \bmod p$ belongs to $E(\F_p)$. Running ECM for $E$ and $N$ consists in computing the multiple $Q=(x_Q:y_Q:z_Q):=[M] P \bmod N$ for $M=(\lfloor C\rfloor!)^{\lfloor \log N/\log 2 \rfloor}$, \emph{i.e.} one uses the chord-and-tangent formul\ae{} and reduces the coordinates modulo $N$ (if two points have distinct coordinates modulo $N$ then one uses the formula for adding two distinct points). We summarize this in Algorithm~\ref{alg:ECM} below.

\begin{algorithm}
\begin{algorithmic}[1]
\REQUIRE parameters $u,v$, an integer $N$, an elliptic curve $E/\Q$ and $P\in E(\Q)$.
\ENSURE a prime factor $p$ of $N$ such that $p<B:=\lfloor N^{1/u}\rfloor$ or FAIL. 
\STATE $C\gets \lfloor B^{1/v} \rfloor$
\STATE $M\gets C!^{\lfloor \log N/\log 2 \rfloor}$
\STATE $Q\colon (x_Q:y_Q:z_Q)\gets [M]P\bmod N$
\STATE $g\gets\gcd(z_Q,N)$ 
\IF{$g\neq 1$} \PRINT $g$
\ENDIF 
\end{algorithmic}
\caption{One curve subroutine of ECM}
\label{alg:ECM}
\end{algorithm}

\begin{claim}\label{claim}
If $|E(\F_p)|$ is $C$-friable for some (unknown) prime factor $p$ of $N$, then $g_N:=\gcd(z_Q,N)$ is a multiple of $p$.
\end{claim}

Let us give the main ideas justifying the claim.
If the points involved in the double-and-add method were all distinct not only modulo $N$ but also modulo $p$, then $Q$ would be the neutral element, so $z_Q\equiv 0\pmod p$. If one used a wrong formula because two points were distinct modulo $N$ but equal modulo $p$, then $x_Q\equiv y_Q\equiv z_Q\equiv 0\pmod p$. In both cases ECM finds a multiple of $p$ and a careful analysis shows that the probability that the result is exactly $p$ is $1-o_{p\to\infty}(1)$ (see~\cite{Len87}). If $g_N$ is a prime factor of $N$ we are done, otherwise we pick a different curve $E$ and start over until a factor is found.  

\medskip
The ECM algorithm consists in repeating Algorithm~\ref{alg:ECM} either once, or a given number of times, or until success, depending on the application. ECM is primarily used to completely factorize $N$, which is done by finding some proper divisor and then iterating. If one takes $u=2$ so that we seek all prime factors less than $B=N^{1/2}$, these prime factors are enough to find a possible cofactor. Next the parameter $v$ is chosen so as to minimize the average running time: $v=\sqrt{2}(\log N)^{1/2}/(\log\log N)^{1/2}$ or equivalently   
$C=\lfloor L_N(1/2,1/\sqrt{2})\rfloor$, where 
\begin{align}
L_N(\alpha,c)= \exp\left(c(\log N)^\alpha(\log \log N)^{1-\alpha}\right),
\end{align}
and one runs Algorithm 1 with as many curves as needed in order to factorize $N$ (note that there is no guarantee that the procedure terminates after testing finitely many curves). The expectation of the number of curves needed in this optimal choice of $v$ is $C^{1+o(1)}$ (see~\cite{Len87}).

If one decides in advance to stop after $B/\psi (B,C)$ curves, then one finds all primes less than $B$ with constant probability (under the heuristics asserting independence relatively to the choice of the elliptic curve). This allows a second application of ECM : given an integer, decide whether it is $B$-friable with no false positive and with a constant proportion of false negatives. 

A third application is as follows : given a large number of integers less than a paremeter $N$ and given a parameter $B=N^{1/u}$, find at least a prescribed proportion $f$ of $B$-friable integers inside the set. This problem is solved by the Number Field Sieve (NFS) algorithm~\cite{NFS1993}, where ECM is used as a building block in the cofactorization step of the relation collection, also called sieving step, (see for example and~\cite[\S 3, page 337]{Kleinjung2014cofactorization}) and in the splitting step  of the discrete logarithm version of NFS (see for example~\cite[\S 4.1, point 1, page  181]{Semaev2006}). In the latter, one needs to find a single $B$-friable integer. To solve the problem, one uses a single elliptic curve or a finite number of them on a large number of integers. For example, the record factorizations of RSA moduli obtained with the CADO-NFS software~\cite{CADO-NFS} use a dozen elliptic curves to test the friability of billions of integers. Also, in \cite[Table 3, page 345]{Kleinjung2014cofactorization} one uses $5$ to $10$ curves whereas the number of integers exceeds half a million per special-q and billions in total. The idea here is to set $v$ so that one has  $\psi(B,B^{1/v})\geq f$
and then, in order to test $B$-friability, one is interested in how many primes up to $B$ are found by Algorithm 1 with entries $v$ and a particular elliptic curve $E$. We are hence interested in the quantity~$\psi_E(x,y)$ for $(x,y)=(B^{1/u},B^{1/uv})$. 

In particular, the heuristics underlying the use of ECM as a friability test states that the larger $\psi_E(x,y)$ gets, the more $y$-friable integers will be found by ECM inside a given set. In other words, $E$ is more ECM-friendly as $\psi_E(x,y)$ increases. We formalize this idea in the following definition.

\begin{definition}[ECM-friendly curves] Let $x>y$ be positive real numbers and  let $E_1/\Q$ and $E_2/\Q$ be two elliptic curves. One says that $E_1$ is more ECM-friendly than $E_2$ with respect to $(x,y)$ if 
$$\psi_{E_1}(x,y) > \psi_{E_2}(x,y)\,.$$  
One says that $E_1$ is more ECM friendly than $E_2$ if there exists $x_0\geq 2$ and a positive valued increasing function $\vartheta$ such that the above inequality holds for all pairs $(x,y)$ such that $x\geq x_0$ and $\vartheta(x)\leq y\leq x$.
\end{definition}
Our main result (Theorem~\ref{th:main}) roughly states that the probability that the number of $\F_p$-points on a given elliptic curve is friable approaches asymptotically the probability for any integer to be friable.
In the case of CM\footnote{For each number field $K$, there is a finite number of CM elliptic curves defined over $K$ but there are overall infinitely many CM elliptic curves.} elliptic curves 
our result is conditional on the Elliott--Halberstam conjecture (EH), an important analytic number theoretic statement about uniformity aspects in the distribution of primes in residue classes. 

\begin{conjecture}[Elliott--Halberstam, {\it e.g.}~{\cite[Hyp B]{Wang2018}
for $\Q$, extended to $K$ quadratic in~\cite[Prop. 2.2]{Pollack2016}}]\label{conjecture:EH}   

Let $K$ be either $\Q$ or an imaginary quadratic field of class number 1. Let $\norm{\cdot}$ denote the field norm relative to $K/\Q$. We define\footnote{Here ``$p\ \dot{\in}\ \OO_K$'' means that we count each prime ideal only once \emph{i.e.} we identify elements that generate the same prime ideal of $\OO_K$.}  
\[
\Pi_K(x)=\{p\ \dot{\in}\ \OO_K, {\rm prime}\colon \norm{p}\leq x\}\,,\qquad
\Pi_K(x;c,a)=\{p \in \Pi_K(x)\colon p\equiv a\bmod c\},
\]
of cardinality denoted $\pi_K(x)$ and $\pi_K(x;c,a)$, respectively. Let $\delta>0$. Then for
any fixed $a\in \OO_K$ and $\omega>0$ we have 
\begin{align*}
 \sum_{
\scaleto{
\begin{array}{c}
\norm{q}\leq x^{1-\delta}\\
(q,a)=1
\end{array}
}{15pt}
 }
 \left|\pi_K(x;q,a)-\frac{\pi_K(x)}{\varphi(q)} \right| \ll_{\omega}  \frac{x}{(\log x)^\omega},
\end{align*} 
for any $x\geq 2$ and where $q\in\OO_K$ and $\varphi(q)=|(\OO_K/q\OO_K)^*|$. \end{conjecture}

The Elliott--Halberstam conjecture is standard in analytic number theory. It is a far reaching generalization of the celebrated Bombieri--Vinogradov Theorem where $K=\Q$ and the bound on $q$ in the index set of the summation cannot exceed $\sqrt{x}$. EH would have countless important applications in number theory. Suffice it to mention the EH bounds contained in~\cite{Maynard2015} on gaps between consecutive primes (or more generally on the length of intervals containing $k$-tuples of primes), as well as~\cite{Zhang2014} where, as a crucial step in the proof of the main result, Zhang establishes a bound towards EH (\emph{i.e.} going beyond the $\sqrt{x}$ threshold) under some extra restrictions on the prime factorization of the moduli involved.


A quantitatively refined version of EH allows one to let $\delta$ depend on $x$ as long as $\delta(x)\to 0$ as $x$ tends to infinity. Indeed, 
H. Montgomery suggested that one could take $\delta(x)\rightarrow 0$ in Conjecture~\ref{conjecture:EH}. Friedlander and Granville~\cite{FriedlanderGranville1992} showed that the conjecture fails if $\delta(x)$ is less than a certain function of~$x$. However Liu, Wu and Xi~\cite{LiuWuXi2019} used the conjecture for $\delta(x)$ large enough not to contradict the necessary constraints observed by Friedlander and Granville (Conjecture~\ref{conjecture:parametricEH} states this ``parametrized'' version of EH).

In order to state our main result, we first recall some classical facts about the counting function of friable integers. With notation as in~\eqref{eq:psi-fct}, one has the well known asymptotics due to Dickman:
$$
\lim_{x\rightarrow\infty}\frac{\psi(x,x^{\frac 1u})}{x}=\rho(u)
\,, 
$$
where $\rho$ is the unique continuous function on $\R_{\geq 0}$ that is differentiable on $(1,\infty)$ and satisfies $\rho(u)\equiv 1$ on $[0,1]$ and $u\rho'(u)=-\rho(u-1)$ on $(1,\infty)$. Asymptotics due to de Bruijn~ (\cite[(1.8)]{debruijn1951}, see also~\cite[Cor. 2.3]{Hildebrand1993}) describe the behaviour of $\rho$ as $u$ grows:
\begin{equation}\label{eq:HT}
\log \rho(u)=-u\Big(\log u+\log_2(u+2)-1+O\big(\frac{\log_2(u+2)}{\log(u+2)}\big)\Big)\qquad (u\geq 1)\,.
\end{equation}

We now state our main result. From Conjecture~\ref{conjecture:EH}, we draw an asymptotic equivalent for $\psi_E(x,y)$. From a refined version of Conjecture~\ref{conjecture:EH} (Conjecture~\ref{conjecture:parametricEH}) we handle uniformity issues in these asymptotics, and finally, we state a non CM analogue (conjectural on Hypothesis~\ref{hypothesis}) of our estimates.

\begin{theorem}\label{th:main}Let $x\geq 2$ and let $y$ satisfy $2\leq y\leq x$. Set $u:=\frac{\log x}{\log y}$. 
\begin{itemize}
    \item Let $E/\Q$ be a CM elliptic curve.
\begin{enumerate}
 \item (Theorem~\ref{th:equivalent})  Assume Conjecture~\ref{conjecture:EH}. If $u$ is upper bounded by an absolute constant then one has
 \[
\psi_E(x,y)\sim \rho(u)\frac{x}{\log x}\qquad (x\to\infty)\,,
\]

\item (Corollary~\ref{cor:main}) Assume Conjecture~\ref{conjecture:parametricEH}. If $\delta(x)$ is a function satisfying for some $\eta>0$ and $\beta>0$:
\[
\frac{\log_2x}{\eta\log x}\leq\delta(x)\ll\frac 1{(\log_2x)^{1+\beta}}
\]
and $u\leq \frac{\log_3x}{\log_4 x}$ then we have
\[
\psi_E(x,y)=\rho(u)\frac{x}{\log x}(1+O(\delta(x) u/\rho(u))).
\]
\end{enumerate}
\item Let $E/\Q$ be a non CM elliptic curve.
\begin{enumerate}[resume]
\item (Theorem~\ref{th:main nonCM}) Assume Hypothesis~\ref{hypothesis}. If $u$ is upper bounded by an absolute constant then one has
\[
\psi_E(x,y)\sim \rho(u)\frac{x}{\log x},\qquad (x\to\infty)\,.
\]
\end{enumerate}
\end{itemize}
\end{theorem}

Note that the same asymptotics hold for $\psi_E(x,y)$ disregarding the endomorphism ring of $E$, provided the relevant assumption is made on $E$ (Conjecture~\ref{conjecture:EH} if $E$ has CM, and Hypothesis~\ref{hypothesis} otherwise). This uniform asymptotic behaviour of elliptic curves over $\Q$ suggests that
ECM-friendliness is determined by the implicit error terms in Theorem~\ref{th:main}. In Section~\ref{ssec:discussion}, we discuss in detail these error terms. In particular, for CM elliptic curves, we show the relevance of introducing as in~\cite[Def. 5.1]{BShinde2021}, the quantity 
\[
\gamma_K=L'(1,\chi)/L(1,\chi),
\] 
where $\chi$ is the Kronecker character of the quadratic field $K$ associated to~$E$. 



We conclude by stating a strong form of Theorem~\ref{th:main}(2) which is directly related to questions in cryptography, as we will discuss in Section~\ref{sec:motivation}.
\begin{theorem}\label{th:Q}
Assume Conjecture~\ref{conjecture:parametricEH}. Let $(x,y,z)$ be three positive integers such that $u:=\frac{\log x}{\log y}$ and $v:=\frac{\log y}{\log z}$ lie in the domain 
\[
\Delta:=\Big\{(u,v)\in\R^2\colon u\leq \frac{\log_2x}{\log_3x}\text{ and }v\leq \frac{\log_3y}{\log_4y}\Big\}\,.
\]
With notation as in~\eqref{eq:psi-fct} we set
\[
\Psi_{E,z}(x,y)=\{n\in \Psi(x,y)\colon \exists p\mid n, |E(\F_p)|\text{ is }z\text{-friable}\}
\]
and we let $\psi_{E,z}(x,y)$ denote its cardinality. Then we have, uniformly on $\Delta$,  
\[
\frac{\psi_{E,z}(x,y)}x=\rho(v)\rho(u)(1+o(1))\qquad (x\to\infty,\,y\to\infty)\,.
\]
\end{theorem}
The paper is organized as follows. In Section~\ref{sec:motivation} we come back to our cryptographic motivation and study the running time of the splitting step of NFS, which is ECM-based, as a consequence of Theorem~\ref{th:Q}. In Section~\ref{sec:background} we prove Theorem~\ref{th:main}(1) following work of Wang~\cite{Wang2018}. In Section~\ref{sec:uniform} we state a uniform version of the Elliott--Halberstam conjecture and, assuming it, we prove Theorem~\ref{th:main}(2). The error terms implicit in Theorem~\ref{th:main} are then discussed and, in the CM case, we prove a computation oriented formula for $\gamma_K$. Section~\ref{sec:Q} is devoted to the proof of Theorem~\ref{th:Q}. Finally in the last section, we investigate the implications, in the non CM case, of heuristics developed by Pollack and we prove Theorem~\ref{th:main}(3).

\medskip{\noindent}
{\bf Acknowledgements.} We thank Jie Wu for suggesting the first author to investigate connections between ECM and the Elliott--Halberstam conjecture, as well as R\'egis de la Bret\`eche, Ofir Gorodetsky and Alessandro Languasco for comments and corrections on an earlier version of the manuscript.

\section{Cryptographic motivation}\label{sec:motivation}

In this section we put Theorem~\ref{th:Q} in context and we give an example of algorithmic application. Precisely Theorem~\ref{th:Q} enables us to perform a computational task (Problem~\ref{problem:ours} below), which is related to a classical problem in cryptography: the splitting step for discrete logarithms (Problem~\ref{problem:splitting} below).
\begin{problem}\label{problem:ours}
Consider a prime $q$, a generator $g$ of $(\Z/q\Z)^*$ and an auxiliary element $h\in (\Z/q\Z)^*$. Let $u$ and $v$ be parameters and let $E/\Q$ be an elliptic curve. Find an integer $e\in[0,q-1]$ such that $n:=g^eh\bmod q$ is $q^{1/u}$-friable and such that, for some prime divisor $p$ of $n$, $E$ has good reduction at $p$ and $|E(\F_p)|$ is $q^{1/(uv)}$-friable. 
\end{problem}

\begin{problem}[Splitting step of NFS]\label{problem:splitting}
Consider the same data as in the problem above. For a parameter~$k$, consider $E_1,E_2,\ldots,E_k$, elliptic curves over $\Q$. Find an integer $e\in[0,q-1]$ such that $n:=g^eh\bmod q$ is $q^{1/u}$-friable and, for all prime factors $p$ of $n$, there exists $i \leq k$ such that $|E_i(\F_p)|$ is $q^{1/(uv)}$-friable.
\end{problem}

 To solve Problem~\ref{problem:ours}, one runs ECM on the integers $g^eh\bmod q$ corresponding to values of $e\in[1,q-1]$ which are  chosen uniformly at random until it is $B$-friable for $B=q^{1/u}$ (see Algorithm~\ref{alg:splitting} for a precise description). Note that the algorithm uses a single CM elliptic curve $E/\Q$ that is required to have positive Mordell--Weil rank. To fix ideas, our description of Algorithm~\ref{alg:splitting} uses, among the 13 possible $j$-invariants of CM elliptic curves defined over $\Q$, the case $j=8000$ for which we have selected one twist of positive rank given by the Weierstrass equation: $E\colon y^2=x^3+x^2-3x+1$ (the point $P=(-1:2:1)\in E(\Q)$ has infinite order).

\begin{algorithm}[ht]
\begin{algorithmic}[1]
\REQUIRE a prime $q$ and two integers $g,h\in[1,q-1]$ such that $g$ is a generator of $\F_q^*$, and two parameters $u$ and $v$
\ENSURE an integer $e$ such that $g^eh\bmod q$ has a factor less than $B=\lfloor q^{1/u} \rfloor$ 
\STATE $E\colon y^2=x^3+x^2-3x+1$, $P=(-1:2:1)\in E(\Q)$ 
\REPEAT 
\STATE $e\gets$ an integer chosen uniformly at random in $[1,q-1]$
\STATE $n\gets g^eh\bmod q$
\STATE run ECM for $n$ and $B$ on the curve $E$, with parameters $u$ and $v$ 
\UNTIL{ECM finds a proper factor of $n$}
\end{algorithmic}
\caption{NFS splitting step}
\label{alg:splitting}
\end{algorithm}

The next statement asserts that Theorem~\ref{th:Q} can be used to solve Problem~\ref{problem:ours}.

\begin{theorem}\label{th:complexity}
\begin{enumerate}
\item Under Conjecture~\ref{conjecture:EH} (resp. Conjecture~\ref{conjecture:parametricEH}), Algorithm~\ref{alg:splitting} solves Problem~\ref{problem:ours} in time $(q^{1/(uv)}/\rho(u)\rho(v))^{1+o(1)}$ (as $q\to\infty$) for bounded $u$ (resp. for $u\leq \frac{\log_3 x}{\log_4 x}$).
\item Assume further that Theorem~\ref{th:Q} can be extended to the domain $(x,y=x^{1/u},z=y^{1/v})$ below :
\[
\Delta':=\Big\{(u,v)\in\R^2\colon u\leq c_u\frac{(\log x)^{1/3}}{(\log_2 x)^{1/3}}\text{ and }v \leq c_v\frac{(\log x)^{1/3}}{(\log_2 x)^{1/3}}\Big\}\,,,  
\]
for two constants $c_u,c_v\geq 3^{1/3}$. Then, with a constant probability, Algorithm~\ref{alg:splitting} on input $q$ terminates in time $L_q(1/3,3^{1/3})^{1+o(1)}$ and solves Problem~\ref{problem:ours} for $u=c_u(\log q/\log_2 q)^{1/3}$ and $v=c_v(\log q/\log_2 q)^{1/3}$.
\end{enumerate}
\end{theorem}

\begin{proof}

1. Recall that $B=\lfloor q^{1/u}\rfloor$ and $C=\lfloor B^{1/v}\rfloor=\lfloor q^{1/(uv)}\rfloor$. As input $N$ of Algorithm~\ref{alg:ECM}, we take the ouput $n$ of Algorithm~\ref{alg:splitting}. The cost of ECM (Algorithm~\ref{alg:ECM}) is essentially that of step $3$, which is $O(\log M)$ by double-and-add exponentiation ($M=C!^{\lfloor \log n/\log 2 \rfloor}$ as defined in Algorithm~\ref{alg:ECM}). By Stirling's formula, this is 
\begin{align}  \label{eq:cost}
\text{time(Alg.~\ref{alg:splitting}: line 5)}=O(\log M)=O(C\log C\log n)=C^{1+o(1)}=q^{1/(uv)+o(1)}\,\, (q\to\infty).
\end{align}

Since $e$ is uniformly chosen at random, the number of executions of the loop in lines 2-6 of Algorithm~\ref{alg:splitting} is, with positive probability, less than a constant times the inverse of the probability of success. We saw earlier ({\it cf.} Claim~\ref{claim}) that the condition in line~$6$ of Algorithm~\ref{alg:splitting} is satisfied if, for a prime factor $p$ of $n$, the order $|E(\F_p)|$ is $C$-friable. We conclude that the number of executions of the loop is $q/\psi_{E,q^{1/(uv)}}(q,q^{1/(u)})$. 

Since $u$ and $v$ are in the domain $\Delta$ defined in Theorem~\ref{th:Q}, we have $q/\psi_{E,q^{1/(uv)}}(q,q^{1/(u)})\leq (1/(\rho(v)\rho(u))(1+o(1))$.
Combining this with~\eqref{eq:cost}, the cost of Algorithm~\ref{alg:splitting} is $(q^{1/(uv)}/\rho(u)\rho(v))^{1+o(1)}$.

2. We set the value of the constants : $c_v=c_u=3^{1/3}$. We inject in~\eqref{eq:HT} the values of $u$ and $v$:
\[
\log(\rho(v)\rho(u))
=(-1+o(1))\cdot(u\log u+v\log v)
=(-1+o(1))\cdot \left(\frac{c_u+c_v}{3}(\log q)^{1/3}(\log_2 q)^{2/3}\right).
\]
Hence the loop is executed at most $L_q(1/3,\frac{c_u+c_v}{3})^{1+o(1)}$ times. Multiplying this by the cost computed in~\eqref{eq:cost}$, C^{1+o(1)}=L_q(1/3,\frac{1}{c_uc_v})^{1+o(1)}$, we find the running time of Algorithm~\ref{alg:splitting}
\begin{align*} 
\text{time(Algorithm~\ref{alg:splitting})}=L_q(1/3,c), 
\end{align*}
where $c=\frac{1}{c_uc_v}+\frac{c_u+c_v}{3}=3^{1/3}$.

\end{proof}

\begin{remark}
One can easily adapt Algorithm~\ref{alg:splitting} to solve Problem~\ref{problem:splitting} (hence the identical names): in line $5$ apply ECM to all the curves $E_i$ with $i=1,\ldots,k$. If we set $k=1/\rho(v)^{1+o(1)}$, and if the outcome of ECM is independent of the input curve $E_i$, then with constant probability we completely factorize $n$ whenever it is $B$-friable. Then we execute the loop in lines 2-6 $q/\psi (q,q^{1/u})=1/\rho(u)^{1+o(1)}$ times. Hence the total cost is $(C/\rho(u)\rho(v))^{1+o(1)}$, which is the same as for Problem~\ref{problem:ours}. To the best of our knowledge, no rigorous argument proves the required ``independence'' property for the input curves at the present time even though a heuristic complexity to solve Problem~\ref{problem:splitting} is well known~(\cite{Semaev2006}).


\end{remark}

\section{Background and proof of Theorem~\ref{th:main}(1)}\label{sec:background}


Let $K$ be either $\Q$ or an imaginary quadratic field of class number $1$. Recall that $\norm{\cdot}$ is the norm map relative to $K/\Q$ and that $P^-(n)$ and $P^+(n)$ respectively denote the smallest and largest prime factors of $n$. Let $a,c\in\OO_K$ and let $\kappa$ be a root of unity of $K$. We set (recall the notation of Conjecture~\ref{conjecture:EH})
\begin{align}\label{eq:Psi_K}
\psi_K(x,y;c,a,\kappa)=|\{ \pi\in \Pi_K(x;c,a)\colon  P^+(\norm{\pi-\kappa})<y\}|.
\end{align}
Deuring's CM theory describes precisely the number of $\F_p$-points on a CM elliptic curve $E/\Q$ having good reduction at $p$. We state an explicit version that can be found in~\cite{RubinSilverberg2009} or~\cite{Cohen2007},
that enables us to relate~\eqref{eq:psiE} with~\eqref{eq:Psi_K}.


\begin{lemma}[{\cite[Th 1.1, Th. 5.3, Th 5.6, Th 5.7]{RubinSilverberg2009} },~{\cite[\S8.5.2]{Cohen2007}}]\label{lemma:CM} For any elliptic curve $E/\Q$ with CM by an order $\OO$ of an imaginary quadratic field $K$, there exists $c\in \OO_K$, a set $A\subset \{a\in \OO_K\colon \gcd(a,c)=1\}$ of cardinality $\varphi(c)/2$ such that for any prime number $p$ not dividing the discriminant of $\OO$ and at which $E$ has good reduction,
\begin{itemize}
    \item if $p$ is inert in $K$ then one has $|E(\F_p)|=p+1$,
\item
if $p$ splits in $K$, there exists $a\in A$ and a root of unity $\mu_{c,a}$ satisfying $|E(\F_p)|=\norm{\pi-\mu_{c,a}}$,  where $\pi\in\OO$ is uniquely determined by the conditions $\norm{\pi}=p$ and $\pi\equiv a \pmod c$.
\end{itemize}
In particular
\begin{align*}
\psi_E(x,y)&=\#\{p\in\Psi_E(x,y)\colon p\text{ inert in }K\}+\sum_{a\in A}\psi_K(x,y;c,a, \mu_{c,a})\\
&=\#\{p\leq x\colon p\nmid \Delta_E,\, p\text{ inert in }K,\, (p+1)\text{ is $y$-friable}\}
+\sum_{a\in A}\psi_K(x,y;c,a, \mu_{c,a})\,.
\end{align*}
\end{lemma}
The lemma is stated in the general case where the elliptic curve $E$ has CM by an unspecified order $\OO$ of an imaginary quadratic field $K$. However the counting functions we study only involve $|E(\F_p)|$, as far as the geometry of $E$ is concerned. Therefore, since there is a canonical $\Q$-rational isogeny $E\to E_0$, where $E_0/\Q$ is an elliptic curve with CM by the full ring of integers $\OO_K$ (see \emph{e.g.}~\cite[Prop. 25]{CCS2013}) we will assume in the sequel that $\OO=\OO_K$ whenever this simplifies the exposition.

To prove Theorem~\ref{th:main}(1), we follow the strategy of Wang~\cite{Wang2018}. In particular we assume Conjecture~\ref{conjecture:EH} and we appeal to the linear sieve of Rosser--Iwaniec that we now recall.


\medskip
Let $\cA\subset \OO_K$ be a finite set, $\cP \subset \OO_K$ a set of primes, $z\geq 2$ a real number and $d\in \OO_K$ a squarefree integer whose prime factors belong to $\cP$. Let $\cA_d=\cA\bigcap d\OO_K$ and $P(z)= \prod_{\norm{p}<z,p\in \cP} p$. Let $X$ be a real number (that should be seen as an ``approximation'' of $|\cA|$) and let $w$ be a multiplicative function on $\OO_K$ such that for $p\in\cP$ one has $0< w(p) <\norm{p}$. 
We set $r(\cA,d)= |\cA_d| -\frac{w(d)}{\norm{d}}X $ (which is expected to be small) and also:
\[
S(\cA;\cP,z)=|\{a\in \cA\colon (a,P(z))=1\}|\,,\qquad 
V(z)= \prod_{p\in \cP,\,\norm p\leq z} (1-\frac{w(p)}{\norm{p}})\,. 
\]

\begin{lemma}[Rosser--Iwaniec~\cite{RosserIwaniec1980}, see also {\cite[Lemme 3.1]{Wang2018}}]\label{lem:Rosser-Iwaniec}
Assume that there exists $\alpha\geq 2$ such that
\begin{equation*}
\prod_{u\leq \norm{p} < v} \left( 1-\frac{w(p)}{p}\right)^{-1} \leq \frac{\log v}{\log u} \left( 1+\frac{\alpha}{\log u} \right)
\end{equation*}
for all $v>u\geq 2$. Then for any $D\geq z\geq 2$  one has
\begin{align*}
S(\cA;\cP,z) \ll X V(z) + \sum_{\norm{d}<D,\, d\mid P(z)} |r(\cA,d)|.
\end{align*}
\end{lemma}
Finally we recollect an estimate for the summatory function of $\mu(n)/\|n\|$ over integers less than $x$ that are not divisible by primes $\leq y$. In the application of Wang's strategy, one of the base steps uses M\"obius inversion, which explains why such summatory functions come into play.


%

\begin{lemma}[{\cite[Lemma 7.2]{delaBretecheFiorilli2020}}\footnote{The first version of this lemma can be found in~\cite{LachandTenenbaum2015} where one has an additional error term $O_\epsilon( \frac{\log(u+1)}{\log y})$. The version of \emph{loc. cit.} suffices for most of our computations, however our discussion of error terms in \S\ref{ssec:discussion} requires the refinement in~\cite{delaBretecheFiorilli2020} .}, generalized to imaginary quadratic fields\footnote{The generalization is direct, hence it is not reproduced here.}]\label{lemma:LT}
 Let $K$ be $\Q$ or an imaginary quadratic field of class number~$1$. Let $\mu$ be the Möbius function generalized to~$K$. For any $\epsilon>0$, we have 
\[
\sum_{
\substack{
\norm{n}\leq x\\
P^-(n)>y
}
}  \frac{\mu(n)}{n}= \rho(u)+O_\epsilon(\exp\{-(\log y)^{\frac35-\epsilon}\})
\]
uniformly in $x\geq 2$ and $\exp\{(\log x)^{\frac25+\epsilon}\}\leq y\leq x$, where $u=\frac{\log x}{\log y}$.
\end{lemma}


We can now recall the statement and give the proof of Theorem~\ref{th:main}(1). 
\begin{theorem}\label{th:equivalent} Let $E/\Q$ be a CM elliptic curve and let $K$ be the associated imaginary quadratic field of class number $1$.
Let $x\geq 2$ and let $y$ be such that $2\leq y\leq x$ and $u:=\frac{\log x}{\log y}$ is upper bounded by an absolute constant. Then we have
\[
\psi_E(x,y)\sim \rho(u)\frac{x}{\log x}\qquad (x\to\infty)\,.
\]
\end{theorem}

\begin{proof}[Proof of Theorem~\ref{th:equivalent}] 


By Lemma~\ref{lemma:CM} we have 
\begin{align}\label{eq:from E to K}
\psi_E(x,y)=&|\{p\text{ split in $K$},\,p\leq x\colon P^+(|E(\F_p)|)<y\}| \nonumber\\
&+|\{p\text{ inert in $K$}, \,p\leq x\colon P^+(|E(\F_p)|)<y\}|
\nonumber\\
             =&\sum_{a\in A }\psi_K(x,y;c,a,\mu_{c,a})+|\{p\text{ inert in $K$}, \,p\leq x\colon P^+(p+1)<y\}|,
\end{align}
where $A$, $c$ and $\mu$ are as in Lemma~\ref{lemma:CM} (see the notation \eqref{eq:Psi_K}). Note that for the purpose of this article one could have added a $O(1)$ term to account for ramified primes in $K$ and for primes of bad reduction of $E$, but one can be more precise and erase the $O(1)$ term because the ramified primes correspond precisely to the primes of bad reduction.
  
The second term of the right hand side is the case $a=-1$ in~\cite[Lemma 4.1]{Wang2018}:  
\begin{align}\label{eq:inert}
|\{p\text{ inert in $K$},\, p\leq x\colon P^+(p+1)<y\}|\sim \rho(u)\frac{\frac{x}{\log x}}{2}. 
\end{align}
We shall prove that $\psi_K(x,y;c,a,\mu)\sim \rho(u)x/(\varphi(c)\log x)$ and, when summing over the $|A|=\varphi(c)/2$ values of $a$ we obtain:
\begin{align}
|\{p\text{ split in $K$},\,p\leq x\colon P^+(|E(\F_p)|)<y\}|\sim \frac{1}{2}\rho(u)\frac{x}{\log x},
\end{align}
which, together with Equation~\eqref{eq:inert} implies the equivalent of $\psi_E(x,y)$ and will complete the proof.

Hence, it remains to prove an equivalent for $\psi_K(x,y;c,a,\kappa)$ for constants $c$, $a\in\OO_K$ and a constant $\kappa\in \OO_K^\times$.


 We note that for large enough $y$, more precisely $y>c$ (which we assume holds in the rest of the proof since $x\to \infty$ and $u$ remains bounded), one has $\gcd(q,c)=1$ as soon as $P^-(\norm{q})>y$. Therefore, by the Chinese Remainder Theorem, we can fix, for each such $q$ an element $a'\in\OO_K$ such that $a'\equiv a\pmod c$ and $a'\equiv \kappa \pmod q$. Combining this with M\"obius' inversion we write 
\begin{align}\label{eq:S1+S2} 
\psi_K(x,y;c,a,\kappa)&= |\{ \pi\in \Pi_K(x;c,a)\colon  P^+(\norm{\pi-\kappa})<y\}| \\
&=|\{ \pi\in \Pi_K(x;c,a)\colon {\rm gcd}\Big(\pi-\kappa,\prod_{\ell\text{ prime} ,\,\|\ell\|\geq y}\ell\Big)=1  \}|\nonumber
\\
&=\sum_{\substack{q\in \OO_K,\,\|q\|\leq x+1\\
P^-(\norm{q})>y}}\mu(q)|\Pi_K(x;c,a)\cap \Pi_K(x;q,\kappa)|
= \sum_{
\substack{
q\in \OO_K,\,\|q\|\leq x+1\\
P^-(\norm{q})>y
}} \mu(q)\pi_K(x;qc,a')\nonumber
\end{align}
where we have used the fact that the algebraic norm of a root of unity is $1$.

In order to evaluate $\psi_K(x,y;c,a,\kappa)$ we follow closely Wang's method~(\cite[D\'em. du Lemme 4.1]{Wang2018}). We highlight the necessary adaptations, omitting the details whenever they are straightforward form Wang's approach. Starting from~\eqref{eq:S1+S2} we fix an arbitrarily small $\delta>0$ and split the counting function $\psi_K(x,y;c,a,\kappa)=S_1+S_2$ where
\begin{equation}\label{eq:S1 and S2}
S_1= \sum_{ \substack{ q\in \OO_K,\, \norm{q} \leq x^{1-\delta}\\P^-(\norm{q})>y } } \mu(q)\pi_K(x;qc,a') \,,\qquad 
S_2=\sum_{ \substack{ q\in \OO_K,\, x+1\geq \norm{q} > x^{1-\delta}\\P^-(\norm{q})>y } } \mu(q)\pi_K(x;qc,a') \,. 
\end{equation}

Next, using the multiplicativity of $\varphi$ and our assumption $y>c$, we further decompose $S_1= S_1'+S_1''$ where
\begin{equation}\label{eq:S1' and S''}
S_1'= \frac{x}{\varphi(c)\log x} \sum_{ \substack{ q\in \OO_K,\, \norm{q} \leq x^{1-\delta}\\P^-(\norm{q})>y } } \frac{\mu(q)}{\varphi(q)} , \,\qquad
S_1''=  \sum_{ \substack{ q\in \OO_K,\, \norm{q} \leq x^{1-\delta}\\P^-(\norm{q})>y } } \mu(q) r(x,qc,a')
\end{equation}
and where $r(x,qc,a')=\pi_K(x;qc,a') - \frac{x}{\varphi(cq)\log x}$.

\medskip
\textbf{Step $1'$.} We show that $S_1'\sim \frac{x}{\varphi(c)\log x}\rho(u)$ as $x\to\infty$ and $u$ remains bounded (under these restrictions we deduce that $S_1'$ is asymptotically larger than a constant times $\frac{x}{\log x}$). For $q$ satisfying $P^-(q)>y$, observe that\footnote{We use the notation $\omega(n)$ for the number of distinct prime factors of an element $n$ in a given PID.}
$$
\frac  1{\varphi(q)}-\frac 1q=\frac 1q\times O\Big(\big(1+\frac 1y\big)^{\omega(q)}-1\Big)=O\Big(\frac 1q\times\frac {\omega(q)}{y}\Big)\,.
$$
Since $\|q\|\leq x^{1-\delta}$, we use the upper bound $\omega(q)\ll \log x$ which we combine with the fact that $y\gg x^\theta$, for some $\theta>0$ (recall that $u$ remains bounded) to conclude that uniformly for any $q$ in the index set of $S_1'$ one has
$$
\frac  1{\varphi(q)}=\frac 1q(1+o(x))\qquad (x\to\infty)\,.
$$
Using Lemma~\ref{lemma:LT}, Wang's computation~\cite[(4.6)]{Wang2018} immediately produces
\begin{equation}\label{eq:S1'}
S_1'=\frac{x}{\varphi(c)\log x}\rho(u)\left(1+O(\delta)+o(x)\right)\sim \frac{\Li(x)}{\varphi(c)}\rho(u)\qquad (x\to\infty, u\ll 1)\,.
\end{equation}

\medskip

\textbf{Step $1''$.} We use Conjecture~\ref{conjecture:EH} with a fixed $\omega>6$ to show that $S_1''=O(x(\log x)^{-1}/(\log x)^{\omega-1})$. Note that this is negligible compared to $S_1'$.

\medskip
Let $0<\tilde{\delta}<\delta$ be so that $\norm{cq}\leq x^{1-\delta}$ whenever $\norm{q}\leq x^{1-\tilde{\delta}}$. Using the triangle inequality and Conjecture~\ref{conjecture:EH} we have
\begin{align}\label{eq:S1''}
S_1''&=  \sum_{ \substack{ q\in \OO_K,\, \norm{q} \leq x^{1-\delta}\\P^-(\norm{q})>y } } \mu(q) r(x,qc,a')\leq \sum_{\norm{q}\leq x^{1-\tilde{\delta}} } r(x,qc,a') \nonumber \\
&\leq  \sum_{\norm{q'}\leq x^{1-\delta} } r(x,q',a')\ll_\omega  \frac{x}{\log(x)^{\omega}}\,.
\end{align}

We now handle the contribution of $S_2$ defined in~\eqref{eq:S1 and S2}. We first apply the triangle inequality and we observe that the primes $p\in \OO_K$ that are counted satisfy $p-\kappa=qcm$ with $m\in\OO_K$ of norm bounded by $\|m\|\leq x^\delta$. Therefore we have
$$
|S_2|\leq \sum_{ \substack{ m\in \OO_K\\ \norm{m} \leq x^{\delta}} } |\{p \text{ prime},\,\|p\|\leq x\colon P^-\big(\big\|\tfrac{p-a'}{mc}\big\|\big)>y\}|\,.
$$
To evaluate this upper bound, the main ingredient used by Wang is the linear sieve.
We apply Lemma~\ref{lem:Rosser-Iwaniec} to
$$
\cA=\cA(m,c,a')=\Big\{\frac{p-a'}{mc}\colon p \in \Pi_K(x;mc,a') \Big\}\,,
\qquad \cP=\{p\text{ prime}\colon p\nmid mca'\}\,, \qquad z\leq y\,.
$$ 
Indeed for this choice of parameters, one has
\begin{equation}\label{eq:majS2}
|S_2|\leq \sum_{ \substack{ m\in \OO_K\\ \norm{m} \leq x^{\delta}} }
S(\cA(m,c,a');\cP,z)\,.
\end{equation}
and therefore, for all squarefree $d\in\OO_K$ having all its prime factors in $\cP$, we have $\gcd(d,m)=1$, and $\cA_d$ is homothetic to the translate $\Pi_K(x;dmc,a')-a'$. In particular $|\cA_d|=\pi_K(x;dmc,a')$. We set  $X=\frac{x}{\varphi(mc)\log x}$ and fix a multiplicative function $w$ on $\OO_K$ satisfying
$$w(p)=\left\{\begin{array}{ll}0, & \text{if }p\mid mca', \\ \frac{\norm{p}}{\norm{p}-1},& \text{otherwise}. \end{array}\right.$$
For any squarefree $d\in\OO_K$ with all its prime factors in $\cP$ we have therefore
$$
\frac{w(d)}{\|d\|}X=\frac{w(d)}{\|d\|}\frac{x}{\varphi(mc)\log x}=\Big(\prod_{p\mid d}\frac 1{\|p\|(1-\tfrac 1{\|p\|})}\Big)\frac{x}{\varphi(mc)\log x}=\frac{x}{\varphi(mcd)\log x}
$$
since $d$ is coprime to $mc$, by definition of $\cP$. 

Finally note that the following inequalities hold for all primes $p\in \OO_K$ with $\norm{p}>2$,
\begin{equation}\label{eq:bracket}
\left(1-\frac{1}{\norm{p}}\right)\geq \left(1-\frac{1}{\norm{p}-1}\right) \geq \left(1-\frac{1}{\norm{p}}\right) \left(1+\frac{1}{\norm{p}^2}\right)\,, 
\end{equation}
therefore, combined with Mertens' formula (see \emph{e.g.}~\cite[Chap. I.6, Th. 1.12]{Tenenbaum-cours}), this shows that the hypotheses of Lemma~\ref{lem:Rosser-Iwaniec} are satisfied.
For any fixed $D\geq z$ we obtain
\[
S(\cA;\cP,z) \ll X \prod_{\norm{p}\leq z}\left(1-\frac{w(p)}{\norm{p}}\right) + \sum_{\norm{d}<D,\, d\mid P(z)} \left||\cA_d|-\frac{x}{\varphi(dmc)\log x}\right|\,.
\]
From this upper bound, combined with~\eqref{eq:majS2}, we deduce that $|S_2|\ll S_2'+S_2''$ where
\begin{equation}\label{eq:S2' and S2''}
S_2'=\sum_{
\norm{m}\leq x^\delta}
 \frac{x}{\varphi(cm)\log x} \prod_{
\substack{
\norm{p}<z\\
p \nmid mca'}
}  \left( 1-\frac{1}{\norm{p}-1} \right)\,,\qquad
S_2''=\sum_{
\norm{m}\leq x^\delta
}  \sum_{\substack{
\norm{d}<D \\
d\mid P(z)
} }  |r(\cA,d)|\,,
\end{equation}
and where $|r(\cA,d)|=|\pi_K(x;dmc,a')- 
\frac{x}{\varphi(dmc)\log x}|$.

We next set $z=D=y^{1-2\delta}$ and recall $\omega>6$. Under these conditions we prove upper bounds for $S_2'$ and $S_2''$.



\medskip
\textbf{Step $2'$.} We prove that $S_2'=O(\Li(x)u\delta)$. 
Since $c$ and $a'$ are constants we relax the condition $p\nmid mca'$ into $p\nmid m$ in the index set of the product appearing in $S_2'$. Then we have 
\begin{align*}
S_2' &\ll_{c,a'} \sum_{
\norm{m}\leq x^\delta
} \frac{x}{\varphi(m)\log x} 
\prod_{\substack{\norm{p}<z\\
p \nmid m}}\Big( 1-\frac{1}{\norm{p}-1} \Big)\\
&\ll \frac{x}{\log x}
\left(\prod_{\norm{p}<z} \Big( 1-\frac{1}{\norm{p}-1} \Big)\right)
\left(\sum_{
\norm{m}\leq x^\delta
} \frac{1}{\varphi(m)} 
\prod_{
\substack{
\norm{p}<z\\
p \mid m}
}  
\Big( 1-\frac{1}{\norm{p}-1} \Big)^{-1}\right)\\
\end{align*}
The first factor over primes on the right hand side is $\ll (\log z)^{-1}=((1-2\delta)\log y)^{-1}$ by Mertens' formula combined with~\eqref{eq:bracket}. For the right-most factor we write $\frac{\varphi(m)}{\norm m}=\prod_{p\mid m}(1-\frac 1{\norm p})$ and note that the function $f$ defined on $\OO_K$ by 
$$
f(m)=
\prod_{p \mid m}
\Big( 1-\frac{1}{\norm{p}-1} \Big)^{-v_p(m)}\Big(1-\frac{1}{\norm{p}}\Big)^{-v_p(m)}
$$
(where $v_p(m)$ is the $p$-adic valuation of $m$) is completely multiplicative. Since in the product defining $f$, each factor at $p$ is $\geq 1$ we obtain:
$$
\sum_{
\norm{m}\leq x^\delta
} \frac{1}{\varphi(m)} 
\prod_{
\substack{
\norm{p}<z\\
p \mid m}
}  
\Big( 1-\frac{1}{\norm{p}-1} \Big)^{-1}\leq 
\sum_{
\norm{m}\leq x^\delta
} \frac{f(m)}{\norm m}\,.
$$
Here note that the general term of the product over primes on the left hand side is $\geq 1$ and therefore the upper bound holds both in the case $z\geq x^\delta$ and $z<x^\delta$. Using a partial Euler product and Mertens' formula combined with~\eqref{eq:bracket} we deduce that
$$
\sum_{
\norm{m}\leq x^\delta
} \frac{1}{\varphi(m)} 
\prod_{
\substack{
\norm{p}<z\\
p \mid m}
}  
\Big( 1-\frac{1}{\norm{p}-1} \Big)^{-1}
\leq \sum_{
\norm{m}\leq x^\delta
} \frac{f(m)}{\norm m}
\ll \prod_{\norm p\leq x^\delta}\Big(1-\frac{1}{\norm{p}-2}\Big)^{-1}\ll \delta\log x\,,
$$
This concludes step $2'$:
\begin{equation}\label{eq:S2'}
S_2'=O\Big(\frac{x}{\log x}u\frac{\delta}{1-2\delta}\Big)=O\Big(\frac{x}{\log x}u\delta\Big)\,.
\end{equation}

\medskip
\textbf{Step $2''$.} We prove that $S_2''=O(x(\log x)^{-1}/\log x^{\omega/2-3})$. Note that if $\norm{m}\leq x^\delta$ and $\norm{d}\leq D=y^{1-2\delta}\leq x^{1-2\delta}$ then $n:=md$ is such that $\norm{n}\leq x^{1-\delta}$. We denote by $\tau(n)$ the number of divisors of $n$. Hence we have, applying Cauchy--Schwarz in the last step,
\[
S_2''\leq \sum_{\norm{n}\leq x^{1-\delta}} \sum_{d\mid n} |r(\cA;cd,a')|
\leq \sum_{\norm{n}\leq x^{1-\delta}} \tau(q) |r(\cA;nc,a')|
\leq  (S_{2,\star}''S_{2,\dagger}'')^\frac12
\]
where 
\[S_{2,\dagger}''= \sum_{\norm{n}\leq x^{1-\delta}}|r(\cA;nc,a')|
\,,\qquad 
S_{2,\star}''=  \sum_{\norm{n}\leq x^{1-\delta}} \tau(n)^2|r(\cA;nc,a')|\,.
\]
For $S_{2,\dagger}''$ we recognize the expression of Conjecture~\ref{conjecture:EH} so, recalling that $c$ is a constant (therefore for big enough $x$ one has $\norm c\leq x^{\delta/2}$), we deduce 
\begin{align*}
S_{2,\dagger}''\ll x/(\log x)^\omega.
\end{align*}
As in~\cite[(4.10)]{Wang2018} we first use a trivial upper bound on   $S_{2,\star}''$: 
\[
 |r(\cA;nc,a)| = \Big| \pi_K(x;nc,a)- \frac{\pi_K(x)}{\varphi(nc)}\Big| \ll  \frac{x}{\norm{n}},
\]
to obtain the upper bound:
$$
S_{2,\star}''\ll x\sum_{\norm{n}\leq x^{1-\delta}} \frac{\tau(n)^2}{\norm{n}}\ll x(\log x)^4 
$$
where the last step uses summation by parts and knowledge of the average order of $\tau^2$ (see \emph{e.g.}~\cite[(1.25)]{Wilson1923}). Note also that, invoking positivity for the general term of the sum, the implied constant is absolute (and in particular does not depend on $\delta$). Overall we obtain
\begin{equation}\label{eq:S2''}
S_2''\ll \Big( \frac{x}{(\log x)^\omega}\Big )^\frac12 x^{\frac 12}(\log x)^{2}\ll \frac{x(\log x)^{-1}}{ (\log x)^{\frac{\omega}2-3}}\,.
\end{equation}
Putting together~\eqref{eq:S1'}, \eqref{eq:S1''}, \eqref{eq:S2'}, and \eqref{eq:S2''}, we see that the sums $S_1''$ and $S_2''$ are negligible compared to $x/\log x$. Since $\rho(u)$ is lower bounded by a constant, the proof of Theorem~\ref{th:equivalent} is finished by letting $\delta\to 0$.
\end{proof}

\section{Uniform version of the Elliott--Halberstam conjecture and proof of Theorem~\ref{th:main}(2)}\label{sec:uniform}

In this section we start by stating a refined version of Conjecture~\ref{conjecture:EH} and we prove Theorem~\ref{th:main}(2) under this refined conjecture. We next discuss the error term implicit in Theorem~\ref{th:main}(1) and (2).

\subsection{Proof of Theorem~\ref{th:main}(2)}

The argument we provide is a consequence of~{\cite[Th. 1.5]{LiuWuXi2019}}. Let us first state the refined version of Conjecture~\ref{conjecture:EH} required in our analysis.

\begin{conjecture}[parametric EH, {\cite[Conj. 1]{LiuWuXi2019}} in the $K=\Q$, same $\delta(x)$ as in the rational case, with the same adjustments as in Conjecture~\ref{conjecture:EH} for the case of quadratic $K$]\label{conjecture:parametricEH} Let $\delta(x)$ be a decreasing function such that
\begin{equation}\label{eq:conjecture 1}
(\log_2x)/(\eta \log x) \leq \delta(x) < \eta\qquad (x\geq x_0(\eta)),
\end{equation}
for any $\eta\in(0,1/2]$
Let $K$ be either $\Q$ or an imaginary quadratic field of class number one. For all $q\in K$ we set $\norm{q}=|\OO_K/q|$ the algebraic norm and $\varphi(q)=|(\OO_K/q)^*|$ the Euler function. Let us consider   
\[
\pi_K(x)=\{p \in \OO_K,\,\text{prime}\colon \norm{p}\leq x\}\,,\,\,
\pi_K(x;c,a)=\{p \in \OO_K,\,\text{prime}\colon \norm{p}\leq x,\,\, p\equiv a\bmod c\}.
\]
Then for any fixed $a\in \OO_K$, $a\neq 0$, and $\omega>0$ we have 
\[
 \sum_{\substack{
q\in \OO_K,\, (q,a)=1 \\
\norm{q}\leq x^{1-\delta(x)}}}
 \left|\pi_K(x;q,a)-\frac{\pi_{K}(x)}{\varphi(q)} \right| \ll_{\omega}  \frac{x}{(\log x)^\omega},
\] 
uniformly for $x\geq x_0(\eta)$. 
\end{conjecture}
The original EH conjecture is stated for $K=\Q$ and constant $\delta$. As already mentioned it is a strengthening of the Bombieri--Vinogradov (BV) Theorem. Huxley~\cite{Huxley1971} proved a number field variant of the BV Theorem (see also~\cite[Cor. p. 203]{Johnson1979} for the particular case of imaginary quadratic fields or Pollack~\cite[Lemma 2.3]{Pollack2016} for imaginary quadratic fields of class number $1$), so it seems natural to state a number field EH conjecture extending Huxley's result with the same range on $q$ as in the original EH conjecture. Finally, Liu et al.~\cite{LiuWuXi2019} extended EH by replacing $\delta$ with a decreasing function. In the present work we use the number field EH. If we only want to focus on non uniform results we can restrict to the original EH. However in order to prove uniform results, our analysis requires a new variant of EH (Conjecture~\ref{conjecture:parametricEH}) which combines the number field EH (extending~\cite[Lemma 2.3]{Pollack2016}) with the parametric EH (see~\cite{LiuWuXi2019}).

The proof of Theorem~\ref{th:main}(2) will follow from a generalized form of~\cite[Th. 1.5]{LiuWuXi2019} that we now state.

\begin{theorem}[{\cite[Th. 1.5]{LiuWuXi2019}}  generalized to imaginary quadratic fields]\label{th:kappa} Let $K$ be an imaginary quadratic field of class number~$1$. Let $a\in\OO_K\setminus\{0\}$ and let $\mu$ denote a root of unity of $K$. Let $\omega>0$ and let $\kappa$ be a non-negative arithmetic function on  $\OO_K$. 
With notation as in~\eqref{eq:Psi_K} and assuming Conjecture~\ref{conjecture:parametricEH} for a given function $\delta$ satisfying~\eqref{eq:conjecture 1}, we have 
\begin{align*}
\sum_{\substack{q\in\OO_K,\,\norm{q}\leq Q \\
(q,a)=1}
}
\left|\psi_K(x,y;q,a,\mu)-\frac{\pi_ K(x)}{\varphi(q)}\rho\left(\frac{\log(x/\norm{q})}{\log y}\right)\right|  \ll_{a,\omega} &\frac{x\sqrt{\sum_{\norm{q}\leq x} \kappa(q)^2/\norm{q}}}{(\log x)^\omega}\\
&+ \pi_{K}(x)\delta(x)u\sum_{\norm{q}\leq Q}\frac{\kappa(q)}{\varphi(q)}
\end{align*}
For every $\epsilon>0$ this upper bound is uniform for $x\geq 2$, $\exp((\log x)^{2/5+\epsilon})\leq y\leq x$ and $Q\leq \min(y,\sqrt{x})$.
\end{theorem}

Note that the original statement~\cite[Th. 1.5]{LiuWuXi2019} is over $\Q$: in \emph{loc. cit.} the prime counting function $\pi_K$ is replaced by the usual prime counting function $\pi$ and $\psi_K(x,y;q,a,\mu)$ is replaced by
$$
\pi(x,y;q,a)=\#\{p\leq x\colon q\mid (p-a),\, P^+(\tfrac{p-a}q)\leq y\}\,.
$$
Obtaining Theorem~\ref{th:kappa} from the original~\cite[Th. 1.5]{LiuWuXi2019} requires minor modifications only. As we will see below, Corollary~\ref{cor:main} (and in turn Theorem~\ref{th:main}(2)) will follow from Theorem~\ref{th:kappa}. If one wants to dispense from generalizing~\cite[Th. 1.5]{LiuWuXi2019} to an imaginary quadratic field, one can instead consider the version of Corollary~\ref{cor:main} proved in Appendix~\ref{appendix:main}, where the arguments used are similar to those of the proof of Theorem~\ref{th:main}(1).

In~\cite[Cor. 1.8 (p. 5)]{LiuWuXi2019} a result is proven for $u\leq \frac{\log_2 x}{\log_3x}$. We note here that a stronger consequence of Theorem~\ref{th:kappa} holds if one restricts to $u\leq \frac{\log_3x}{\log_4x}$, and implies in turn Theorem~\ref{th:main}(2).

\begin{corollary}\label{cor:main}
Let $K$ be $\Q$ or an imaginary quadratic field of class number~$1$. Let $a, c\in\OO_K$ with $a\neq 0$ and let $\mu$ be a unit of $K$. Finally let $\beta>0$.
Assuming Conjecture~\ref{conjecture:parametricEH}, we have 
\begin{align}\label{eq:corollary:delta}
\pi_K(x,y;c,a,{\mu})=\frac{x}{\varphi(c)\log x}\rho(u)  \left(1+O\left( \delta u {\rho(u)^{-1}} ) \right)\right). 
\end{align}
uniformly for $1\leq u\leq \frac{\log_3x}{\log_4 x}$ and for $\delta(x)\ll \frac1{(\log_2 x)^{1+\beta}}$.
Consequently Theorem~\ref{th:main}(2) holds.

\end{corollary}

\begin{proof}
First note that under the stated assumptions $(\delta(x)u/\rho(u))=o(1)$ as $x\to \infty$. Indeed, by~\eqref{eq:HT}, one has $\rho(u)\gg \exp(-(1+\beta')u\log u)$ for $u\geq 1$ and any fixed $\beta'$ satisfying $0<\beta'<\beta$. Therefore we compute:
\begin{align*}
\frac u{\rho(u)}&\ll u{\rm e}^{(1+\beta')u\log u}\leq\frac{\log_3x}{\log_4x}
\exp\Big((1+\beta')\frac{\log_3 x}{\log_4x}\log\big(\frac{\log_3 x}{\log_4 x}\big)\Big)\\
&\leq \frac{\log_3x}{\log_4x}
\exp\Bigg((1+\beta')\log_3x\Big(1-\frac{\log_5x}{\log_4x}\Big)\Bigg)\\
&\leq (\log_2 x)^{1+\beta'}\frac{\log_3x}{\log_4x}
\exp\Big(-(1+\beta')\frac{\log_3x}{\log_4x}\Big)=o\big((\log_2 x)^{1+\beta}\big)\,.
\end{align*}

Next the assumption on the size of $u$ implies that $\log y\geq \log x\log_4x/\log_3x$ and thus for big enough $x$ one has $c\leq\min(y,\sqrt x)$, since $c$ is fixed. Setting $\kappa=\textbf{1}_{\{c\}}$ in Theorem~\ref{th:kappa} 
we obtain:
\begin{equation}\label{eq:deltau}
\left|\psi_K(x,y;c,a,\mu)-\frac{\pi_K(x)}{\varphi(c)}\rho\left(\frac{\log(x/\norm{c})}{\log y}\right)\right|  \ll_{a,\omega} \frac{x}{\sqrt{ \norm{c}}(\log x)^\omega}+ \frac{\pi_K(x)}{\varphi(c)}\delta(x)u.
\end{equation}
The second summand on the right hand side of~\eqref{eq:deltau} is $\ll \frac{x}{\log x}\delta(x)u$. Likewise, since $u\geq 1$, and since $\delta$ is lower bounded by assumption in Conjecture~\ref{conjecture:parametricEH}, we have for the first summand:
$$
\frac{x}{(\log x)^{\omega}}\ll_\eta\frac{x}{\log x}\delta(x)u \frac{\log x}{(\log x)^{\omega-1}\log_2 x}
$$
which is $\ll \frac{x}{\log x}\delta(x)u$ for any fixed $\omega\geq 2$.

 

Finally, since $\rho$ is smooth on $(1,\infty)$, $c$ is fixed, and we may assume that $y$ is big enough (recall that the assumptions imply that $\log y\geq \log x\log_4x/\log_3x$) there exists  $\xi \in (u- \frac{\log \norm{c}}{\log y},u)$ such that
\[
\left|\rho\big(\frac{\log(x/\norm{c})}{\log y} \big)-\rho(u)\right|=\left|\frac{\log \norm{c}}{\log y}\rho'(\xi)\right|\leq \left|\frac{\log \norm{c}}{\log y}\frac{\rho(\xi-1)}\xi \right|\ll 
\left| \frac{\rho(\xi-1)}{\log x}\right|
\,.\]
We deduce
\[
\frac 1{x(\log x)^{-1}\delta(x)u}\left|\rho\big(\frac{\log(x/\norm{c})}{\log y} \big)-\rho(u)\right|\ll_\eta \frac 1x\frac{\log x}{\log_2 x}=o(1)\,.
\]
This finishes the proof of~\eqref{eq:corollary:delta}.

To deduce Theorem~\ref{th:main}(2), we combine~\eqref{eq:from E to K} with~\eqref{eq:corollary:delta}, using again that $c$ depends only on $E$. 

\end{proof}


\subsection{Discussion on the implicit error terms in Theorem~\ref{th:main}}\label{ssec:discussion}

The error term in the estimates of Theorem~\ref{th:main} plays an important role in deciding whether an elliptic curve $E_1$ is more ECM-friendly than a second curve $E_2$. To explain this, let us first recall~\cite[Problem 5.1]{BShinde2021}.

\begin{problem}\label{prob:3}
Let $E/\Q$ be an elliptic curve without CM. Decide whether there exists a real number $\beta(E)$ such that
\begin{align*}
\Prob(\# E(\F_p)\text{ is $B$-friable}\colon p\sim n) \sim_n \Prob(m\text{ is $B$-friable}\colon m\sim n{\rm e}^{\beta(E)}),
\end{align*}
where $\sim_n$ denotes the asymptotic equivalent as $n\to\infty$, for positive numbers $a,b$ we write $a\sim b$ as a shorthand for $a\in[b-\sqrt{b},b+\sqrt{b}]$, and ``$\Prob$'' on the left hand side denotes the natural density of a subset of primes, while ``$\Prob$'' on the right hand side denotes the uniform probability on a finite set. 
\end{problem}

Next we mention two results that investigate the size of the error terms in approximations of the counting function of friable integers. 

\begin{theorem}[{\cite[Cor. 1.2, Th. 1.3]{Scourfield2004}}]\label{th:Scourfield}
Let $K$ be an imaginary quadratic field. Then for a fixed $ \varepsilon >0$, for all $x$ and $y$ such that $(\log_2 x)^{5/3+\varepsilon}\leq \log y\leq \log x$, one has 
\[
\psi_K(x,y)=L(1,\chi)x\rho(u)\left(1-\tfrac{\xi(u)}{\log y}\Big(\gamma_K+O\big(\tfrac{\log u}{\log y}+\tfrac{\log u}{\sqrt u}\big)\Big)\right)\qquad (u\to\infty)\,.
\]
Here $\psi_K(x,y)=|\{(a)\text{ ideal of }\OO_K\colon \|a\|\leq x,\,\max\{\mathcal N\mathfrak p\colon\mathfrak p\triangleleft\OO_K\text{ prime},\,\mathfrak p\mid (a)\}\leq y\}|$, 
$\gamma_K=(L'/L)(1,\chi)$ for $\chi$ the Kronecker character of~$K$, and $\xi(u)$ is defined for $u>1$ by the equality $\exp(\xi(u))=1+u\xi(u)$.

In particular, since there are $L(1,\chi)x(1+o(1))$ integral ideals of $\OO_K$ of norm $\leq x$, one has
$$\frac{\psi_K(x,y)}{\psi_K(x,\infty)}=\rho(u)\left(1-\frac{\log(u+1)}{\log y}\gamma_K(1+o(1))\right)\,.$$
\end{theorem}

The result was generalized from $\zeta_K$ to a large class of Dirichlet series of the form $Z(s) G(s)$ where $Z$ is a product of zeta functions with positive exponents and $G$ a well behaved function (\emph{e.g.} a holomorphic function). The following particular case is sufficient for our applications.

\begin{theorem} [{\cite[Th. 1.1, Th. 1.2]{Hanrot2008}}, case $Z=\zeta$, $G$ holomorphic]\label{th:Hanrot}
Let $h$ be an arithmetic function with Dirichlet series $\cH(s)=\sum_{n\geq 1} \frac{h(n)}{n^s}$. We assume that $\cH$ extends to a meromorphic function with a simple pole at~$s=1$ and we write $\cH(s)=a_0/(s-1)+a_1+O(s-1)$ in a neighborhood of $1$.
Then one  has
\[
\sum_{\substack{n\leq x \\P^+(n)\leq y}}h(n)=x\rho(u)\left(a_0+a_1\frac{\log(u+1)}{\log y}+O\left(\frac{(\log(u+1))^2}{(\log y)^2}\right)\right),
\]
uniformly on $\exp((\log_2 x)^{5/3+\epsilon})\leq y\leq x$, for any fixed~$\epsilon>0$.
\end{theorem}
Let us add that a similar result holds for $Z=1/\zeta$, the Dirichlet series of $\mu$. 
The case $h=1$, $a_0=1$, $a_1=\gamma-1$ (the Euler--Mascheroni constant) of Theorem~\ref{th:Hanrot}  was previously established by Saias~\cite[Main corollary]{Saias1989} and yields in particular
\begin{equation}\label{eq:devPsi}\psi(x,y)=x\rho(u)\left(1+\Big(\frac{\log (u+1)}{\log y}(\gamma-1+o(1))\Big)\right)\,,
\end{equation}
as $x\to\infty$ and under the same restrictions on $(x,y)$ as in Theorem~\ref{th:Hanrot}.

Note that Theorem~\ref{th:main}(1) gives a positive answer to Problem~\ref{prob:3} in the CM case while Theorem~\ref{th:main}(3) does so in the non CM case. 
Finally, Theorem~\ref{th:main}(2) raises the necessity of finding asymptotics for $\log(\psi_E(x,y)/\psi(x,y))$. The numerical statistics in Appendix~\ref{sec:numerical} suggest that the following question is relevant.
\begin{problem}\label{problem:main}
Let $E$ be a CM elliptic curve and let $K$ be the associated imaginary quadratic field.  Let $\chi$ be the Kronecker character of $K$ and let $\gamma_K=L'(1,\chi)/L(1,\chi)$. Does the following formula hold:
 \[
 \log(\psi_E(x,y)/\psi(x,y))\sim (-\gamma+1-\gamma_K) \frac{\log(u+1)}{\log y}~\text{?}
  \]
\end{problem}
\begin{remark}
The result~\cite[Th.  1.1]{LachandTenenbaum2015}, which was used in~\cite{Wang2018} and is sufficient for Theorem~\ref{th:main}(1), is not enough here because the error term given is $O(\log(u+1)/\log y)$. We use instead the stronger Lemma~\ref{lemma:LT} due to de la Bretèche and Fiorilli.
\end{remark}

If Problem~\ref{problem:main} receives a positive answer, 
the constant $\gamma_K$ will be used as a criterion to compare ECM-friendliness of CM elliptic curves.  In the non CM case, Peter Montgomery used without proof\footnote{Peter Montgomery is famous for having invented algorithms and concepts which are very effective in computer science but are not justified rigorously or are not presented as part of a broader theory. For instance the modern presentation~\cite{MontgomeryMultiplication} of the Montgomery reduction is Barrett's reduction with $\Q_2$ replacing $\R$ whereas the use of Murphy's $\alpha$ to compare polynomials for NFS, originally used by Montgomery, was justified in~\cite{BLachand2017}.}~(\cite[\S 6.3, pp. 75--76]{MontgomeryPhD}) the constant $\alpha(E)=\sum_\ell\alpha_\ell(E)$ (see Proposition~\ref{prop:alpha} below) to compare the ECM-friendliness of two given elliptic curves $E_1$ and $E_2$, where the sum is over primes $\ell$ such that $\alpha_\ell(E_1)\neq \alpha_\ell(E_2)$. The next result recalls~\cite[Th. 5.1]{BShinde2021} which justifies the existence of $\alpha(E)$ in the non CM case, and gives an analogue of $\alpha(E)$ in the CM case. Moreover we relate explicitly the quantities $\gamma_K$ and $\alpha(E)$ in the CM case.

\begin{proposition}\label{prop:alpha}
Let $E/\Q$ be an elliptic curve. For every rational prime $\ell$ we set 
\begin{itemize}
    \item if $E$ is a non CM curve,
    \[
\alpha_\ell(E)=
\Big(\frac{1}{\ell-1}-\EE_p(\val_\ell(|E(\F_p)|))\Big)\log \ell\,,
\]
\item If $E$ is a CM curve,
\[
\alpha_\ell(E)=
\Big(\frac 3{\ell-1}-4\EE_p(\val_\ell(|E(\F_p)|))\Big)\log\ell
\]
\end{itemize}
where $\EE_p$ is the operator $\lim_{x\to\infty}\pi(x)^{-1}\sum_{p\leq x,\, p\nmid\Delta_E}$ and $\val_\ell$ denotes the $\ell$-valuation. 

Then, the series $(\sum_\ell \alpha_\ell(E))$ converges. Furthermore, denote by $\alpha(E)$ the limit of the converging series: if $E/\Q$ has CM by an order of an imaginary quadratic field $K$, one has the formula
\begin{equation} \label{eq:compute alpha}
\alpha(E)=\gamma_K-\Sigma_K\,,
\end{equation}
where $\gamma_K=L'(1,\chi)/L(1,\chi)$ for $\chi$ the Kronecker character of $K$, and where $\Sigma_K$ is the value of the following converging sum depending only on $K$:
\[
\Sigma_K=\sum_{\ell\text{ inert}}\frac {2\log \ell}{\ell^2-1}\Big(-1+\frac{ 2\ell^2}{\ell^2-1}\Big)
+ \sum_{\ell\text{ ram.}}\frac{\ell\log\ell}{(\ell-1)^2}
+\sum_{\ell\text{ prime}}\log \ell\frac{3+\chi(\ell)}{(\ell-1)^2}\,.
\]
\end{proposition}

\begin{proof}
As already mentioned, the non CM case is due to Barbulescu--Shinde~\cite[Th. 5.1]{BShinde2021}.

Consider an elliptic curve $E/\Q$ that has CM by an order of an imaginary quadratic field $K$.
Fix $s\in \C$ such that ${\rm Re}(s)>1$. We use the factorization $\zeta_K(s)=\zeta(s)L(s,\chi)$ of the Dedekind Zeta function $\zeta_K$ of $K$ combined with the fact that the logarithmic derivative of $\zeta_K$ at $s$ coincides, up to sign, with the Dirichlet series at $s$ of the von Mangoldt function of $K$. We obtain
\begin{align*}
    \frac{\zeta_K'(s)}{\zeta_K(s)}&=\frac{L'(s,\chi)}{L(s,\chi)}+\frac{\zeta'(s)}{\zeta(s)}
    =-\sum_{k\geq 1,\, \mathfrak p}\frac{\log (\mathcal N\mathfrak p)}{(\mathcal N\mathfrak p)^{ks}}\\
    &=-\sum_{\ell\text{ prime}}\frac{(1+\chi(\ell))\log \ell}{\ell^s-1}
    -\sum_{\substack{\ell\text{ prime}\\ \text{unram. in } K }}\frac{(1-\chi(\ell))\log \ell^{2}}{2(\ell^{2s}-1)}\,,
\end{align*}
where, in the first sum, $\mathfrak p$ runs over the prime ideals of $\OO_K$ and $\mathcal N \mathfrak p=|\OO_K/\mathfrak p|$. 
Using the analogous link between the logarithmic derivative of $\zeta$ and the classical von Mangoldt function, we deduce that
\[
\frac{L'(s,\chi)}{L(s,\chi)}=
\sum_{\ell\text{ prime}}\log \ell\Big(\frac{1}{\ell^s-1}-  
\frac{1+\chi(\ell)}{\ell^s-1}\Big)-
\sum_{\substack{\ell\text{ prime}\\ \text{unram. in } K }}
\log \ell
\frac{1-\chi(\ell)}{\ell^{2s}-1}  \,.
\]
Since both sums on the right hand side converge at $s=1$, we let $s\to 1$ and get
\begin{equation}\label{eq:L'/L}
\frac{L'(1,\chi)}{L(1,\chi)}=-
\sum_{\ell\text{ prime}}\log \ell\Big(\frac{\chi(\ell)}{\ell-1}+\frac{|\chi(\ell)|(1-\chi(\ell))}{\ell^2-1}\Big)\,.
\end{equation} 
Next we fix any prime number $\ell$ and compute
\[
\EE_p(\val_\ell(|E(\F_p)|))=\lim_{x\to\infty}\frac 1{\pi(x)}\sum_{\substack{p\leq x\\p\text{ good}}}\val_\ell(|E(\F_p)|)
\]
Let $c\in\OO_K$ be as in Lemma~\ref{lemma:CM}. Let $R$ denote either $\Z$ or $\OO_K$. For any prime $\lambda\in R$ and for any integer $k\geq 0$, the density of primes $\pi\in R$ that\footnote{Again, we identify prime ideals with one given generator.} satisfy $\pi\equiv\mu+b\lambda^k\bmod \lambda^{k+1}$, for a fixed unit $\mu$ of $R$ and some $b\in (R/(\lambda))^*$ (\emph{i.e.} primes $\pi$ for which $\pi-\mu$ has $\lambda$-adic valuation equal to $k$) is $|(R/(\lambda))^*|/|(R/(\lambda^k))^*|=\varphi(\lambda)/\varphi(\lambda^{k+1})=\|\lambda\|^{-k}$, where $\varphi$ is Euler's indicator function for $R$. Moreover, if we require the extra condition $\pi\equiv a\bmod c$, the Chinese Remainder Theorem asserts that the density of the primes $\pi$ considered shrinks to  $\varphi(\lambda)/(\varphi(\lambda^{k+1})\varphi(c))$, which equals $\varphi(c)^{-1}\|\lambda\|^{-k}$. Using this density computation (combined with Lemma~\ref{lemma:CM}) in the case $R=\Z$ we obtain the contribution of inert primes to  $\EE_p(\val_\ell(|E(\F_p)|))$:
\[ \frac 1{\pi(x)}\sum_{\substack{p\leq x\\ \text{inert in }K}}{\rm val}_\ell(p+1)=\frac 12\sum_{k\geq 0}k\frac{\#\{p\leq x\colon {\rm val}_\ell(p+1)=k\}}{\pi(x)}
\rightarrow \frac 12\sum_{k\geq 0}\frac{k}{\ell^k}=\frac 12\frac\ell{(\ell-1)^2}\,.
\]
In the case $R=\OO_K$ we handle, using Lemma~\ref{lemma:CM} again, the contribution of split primes. To do so we use the notation of Lemma~\ref{lemma:CM} and factorize $\pi-\mu_{c,a}=\prod_i\lambda_i^{e_i}$, where $\lambda_i$ is a prime of $\OO_K$ above a prime number $\ell_i$. If $\ell_i$ is inert in $K$ then $\val_{\ell_i}(\|\pi-\mu_{c,a}\|)=2e_i$, otherwise $\val_{\ell_i}(\|\pi-\mu_{c,a}\|)=e_i$. Thus, using similar computations as the ones performed in the case of inert primes $p$, we obtain:
\begin{align*}
\lim_{x\to\infty}\frac 1{\pi(x)}\sum_{\substack{p\leq x\\ \text{split in }K}}{\val}_\ell(|E(\F_p)|)&= \lim_{x\to\infty}\frac 1{\pi(x)}\sum_{a\in A}\sum_{\substack{\pi\in\OO_K,\, \|\pi\|=p\leq x\\ \pi\equiv a\bmod c}}\val_\ell(\|\pi-\mu_{c,a}\|)\\
&= \frac {|A|}{\varphi(c)}\Big(\frac{1+\chi(\ell)+{\bf 1}_{\ell\mid{\rm disc }K}}2\frac{\ell}{(\ell-1)^2}+\frac{|\chi(\ell)|(1-\chi(\ell))}2\frac{2\ell^2}{(\ell^2-1)^2}\Big)\\
&=\frac{1+\chi(\ell)+{\bf 1}_{\ell\mid{\rm disc }K}}4\frac{\ell}{(\ell-1)^2}+\frac{|\chi(\ell)|(1-\chi(\ell))}4\frac{2\ell^2}{(\ell^2-1)^2}
\end{align*}
Overall
\begin{align}\label{eq:alpha theoretical}
4\EE_p(\val_\ell(|E(\F_p)|))=\frac{3+\chi(\ell)}{\ell-1}\Big(1+\frac 1{\ell-1}\Big)+\frac{{\bf 1}_{\ell\mid{\rm disc}K}\ell}{(\ell-1)^2}+\frac{4{\bf 1}_{\ell\text{ inert}}\ell^2}{(\ell^2-1)^2}
\end{align}
Combining with~\eqref{eq:L'/L}, one deduces as wished,
\begin{align*}
\sum_{\ell\text{ prime}}
\Big(4\EE_p(\val_\ell(|E(\F_p)|))-\frac3{\ell-1}\Big)\log\ell=& -\frac{L'(1,\chi)}{L(1,\chi)}+\sum_{\ell\text{ inert}}\frac {2\log \ell}{\ell^2-1}\Big(-1+\frac{ 2\ell^2}{\ell^2-1}\Big)\\
&+ \sum_{\ell\text{ ram.}}\frac{\ell\log\ell}{(\ell-1)^2}
+\sum_{\ell\text{ prime}}\frac{3+\chi(\ell)}{(\ell-1)^2}
\end{align*}

\end{proof}

\begin{example} We have computed the value of $\gamma_K=L'/L(1,\chi)$, $\Sigma_K$ and respectively $\alpha(E)$ using Equation~\eqref{eq:L'/L}, the  summed for $\ell\leq 10^6$. The rapidly converging series $\Sigma_K$ is evaluated using the formula in the statement of Porposition~\ref{prop:alpha} using $\ell\leq 10^6$. Finally, for each prime $\ell\leq 10^4$ one approximates the average value of $\val_\ell |E(\F_p)|$ using the primes $p\leq 10^3$; and we obtained $\widetilde{\alpha}(E)$. The results illustrate the equality $\widetilde{\alpha}(E)\approx \alpha(E)=\gamma_K-\Sigma_K$ for a list of elliptic curves having CM by the quadratic fields $K=\Q(\sqrt{-d})$, $(d>0$) of class number $1$. 

\begin{small}
\begin{align*}
\begin{array}{c|c|c|c|c|c|c|c|c|c}
d   & 
1     & 2     & 3     & 7     & 11    & 19    & 43    & 67   & 163    \\
\hline
\widetilde{\alpha}(E)     & 
-3.042 &    -2.990 &   -3.038 &   -3.073 &   -3.019 &   -3.045 &   -3.080 &   -3.091 &  -3.119 \\
\alpha(E) &
-2.268 & -3.058 &
-1.878 & -3.924 &
-2.908 & -2.284 &
-1.541 & -1.041 &
0.585 
\\
 \Sigma_K  & 
 2.509 & 3.032 &
2.242 & 3.936 &
2.820 & 2.194 &
1.793 & 1.692 & 1.594 \\
\gamma_K      &
0.245 & -0.022 &
0.367 & -0.015 &
-0.085 & -0.085 &
0.246 & 0.659 &
2.171 \\
\hline
\widetilde{\alpha}(E)-(\gamma_K-\Sigma_K) &
0.78 &   0.07 &  -1.16 &  0.85 & -0.11 &  -0.76 &  -1.54 &   -2.05 &   -2.53
\end{array}
\end{align*}
\end{small}
The difference $\widetilde{\alpha}(E)-(\gamma_K-\Sigma_K)$ is close to $0$, but not negligible, depending on $K$. 
Indeed the numerical estimation of the average $\ell$-valuation of $|E(\F_p)|$ is slow and the sample of primes $p\leq 10^3$ is not sufficient to produce very small differences $\widetilde{\alpha}(E)-(\gamma_K-\Sigma_K)$. Note also that the rate of convergence of the series involved seems to depend on the field $K=\Q(\sqrt{-d})$.
\end{example}

\begin{remark}
The computation of $L'(1,\chi)/L(1,\chi)$ is slow if one uses a naive evaluation of each of the series $L'(s,\chi)$ and $L(s,\chi)$ (see~\cite{languasco2021} for a recent algorithm). This gives a second purpose of the formula~\eqref{eq:compute alpha}: quickly computing $(L'/L)(1,\chi)$. Note that~\cite{BLachand2017} gives bounds on the convergence speed. 
\end{remark}

\begin{remark}
In the study of friability of binary forms, Murphy~\cite{Murphy1998} associated a function to irreducible polynomials $f\in\Z[x]$ as follows. For a prime $\ell$,
\begin{align*} 
\alpha_\ell(f)&= (\log \ell)\cdot\Big(\Esp_n(\val_\ell n) -  \Esp_{(a,b)=1} \big(\val_\ell b^{\deg(f)} f(a/b)\big)\Big),\\
\alpha(f)&= \sum_{\ell\text{ prime}} \alpha_\ell(f),   
\end{align*}
where $\Esp_{(a,b)=1}$ corresponds to natural density for randomly chosen pairs of integers $(a,b)$ which are relatively prime; the convergence of the series is proven in~\cite[\S 2.2 ]{BLachand2017}. We note that $\alpha(E)$ has an expression similar to $\alpha(f)$ with $f$ such that $K\simeq \Q[x]\slash (f)$ (up to the condition $(a,b)=1$).
\end{remark}


\section{The set $\Psi_{E,z}(x,y)$: proof of Theorem~\ref{th:Q}}\label{sec:Q}
This section is devoted to the proof of Theorem~\ref{th:Q}. We first state and prove a lemma, which is a variation on the fact that a set of primes which has a natural density also has an analytic (or logarithmic) density (see~\cite[\S III.1]{Tenenbaum-cours}).


\begin{lemma}\label{lemma:density}
Let $Q$ be a set of primes and, for $x\geq 2$, let $\Pi_Q(x)=Q\cap [1,x]$. In the case where $Q$ is the set of all primes, we will simply write $\Pi(x)$ for $\Pi_Q(x)$. Assume that there exists a positive non increasing function $\lambda(x)$ and a constant $\omega>0$ such that for all $x\geq 2$, $$\frac{|\Pi_Q(x)|}{|\Pi(x)|}-\lambda(x)\ll\frac{1}{\big(\log_2 x)^{1+\omega}}\,.$$
Then we have 
\[
\frac{\sum_{p\in \Pi_Q(x)}p^{-1}}{\sum_{p\in \Pi(x)}p^{-1}}
 = \lambda(x)(1+o(1))
 +O\Big(\frac{1}{(\log_2 x)^{1+\omega}}\Big)\,.
\]
\end{lemma}
\begin{proof}
First note that
\[
\sum_{p\in \Pi_Q(x)}p^{-1}
=
\sum_{n=1}^{ \lfloor x \rfloor } \frac{|\Pi_Q(n)|-|\Pi_Q(n-1)|}{n}\,.
\]
An Abel summation then yields
\[
\sum_{n=1}^{ \lfloor x \rfloor  } \frac{|\Pi_Q(n)|-|\Pi_Q(n-1)|}{n}
=
\frac{|\Pi_Q(x)|}{\lfloor x\rfloor+1}+\sum_{n=1}^{ \lfloor x \rfloor } \frac{|\Pi_Q(n)|}{n(n+1)}\,.
\]

Now we use the Prime Number Theorem under the form $|\Pi(x)|=(x/\log x)(1+o(1))$. We obtain: 
\begin{equation}\label{eq:PiQ}
\sum_{n=1}^{ \lfloor x \rfloor } \frac{|\Pi_Q(n)|}{n(n+1)}
= 
 \sum_{n\leq x} \lambda(n)\Big(\frac{1}{n\log n}+o(n^{-2})\Big)+O\Big(\sum_{n\leq x} \frac{1}{n\log n(\log_2 n)^{\omega+1}}\Big)\,.
\end{equation}
To handle the error term we make use of Cauchy's condensation criterion. Precisely
$$
\sum_{n\leq x} \frac{1}{n\log n(\log_2 n)^{\omega+1}}\ll \sum_{1\leq 2^k\leq x}\frac{2^k}{2^k\log(2^k)\log_2(2^k)^{\omega+1}}\ll \sum_{k\leq \log x}\frac{1}{k(\log k)^{\omega+1}}\ll (\log_2(x))^{-\omega}\,.
$$
Plugging this into~\eqref{eq:PiQ} and using the fact that $\lambda(n)\geq \lambda(x)$ for all $n\leq x$, one deduces that
\begin{align*}
\sum_{n=1}^{ \lfloor x \rfloor } \frac{|\Pi_Q(n)|}{n(n+1)}
&= 
\lambda(x) \Big(\sum_{n\leq x} \frac{1}{n\log n}+O(1)\Big)+O\big((\log_2x)^{-\omega}\big)\\ 
&=\lambda(x) \log_2x(1+o(1))+O\big(\lambda(x)+ (\log_2x)^{-\omega}\big). 
\end{align*}
Finally, the term $|\Pi_Q(x)|/\lfloor x\rfloor$ has size $\frac{\lambda(x)}{\log x}(1+o(1))+O(\frac{1}{\log x(\log_2 x)^{1+\omega}})$, which is negligible compared to the left hand side of~\eqref{eq:PiQ}. The proof of the lemma is finished by using Dirichlet's estimate $\sum_{p\in\Pi(x)}p^{-1}=\log_2 x+O(1)$.
\end{proof}

\begin{proof}[Proof of Theorem~\ref{th:Q}]
Let $Q=\{p\text{ prime }\colon |E(\F_p)|\text{ is }z\text{-friable}\}$ and recall that $z=y^{1/v}$. 
\begin{align*}
\frac{\psi_{E,z}(x,y)}{\psi(x,y)}&= \Big(\sum_{
p\in Q,\,
p\leq y
}\psi(x/p,y)\Big)\cdot\Big(\sum_{p\leq y}\psi(x/p,y)\Big)^{-1}\\
 &=
\Big(\sum_{
p\in Q,\,
p\leq y
}  \frac xp \rho(u)(1+\varepsilon(x,y,p))  \Big)\cdot\Big(\sum_{p\leq y}  \frac xp\rho(u)(1+\varepsilon(x,y,p))   \Big)^{-1}\,,
\end{align*}
where $\varepsilon(x,y,p)=(\psi(x/p,y)-(x/p)\rho(u))/((x/p)\rho(u))$. 
We combine~\eqref{eq:devPsi} with the fact that for any $u\in\Delta$, one has $\frac{\log(u+1)}{\log y}\ll \frac{\log_2x}{\log x}$ to obtain that $\varepsilon(x,y,p)= O(\frac{\log (u+1)}{\log y})=O(\frac{1}{(\log x)^\omega})$ for any fixed $0<\omega<1$ and for $u\in \Delta$. We deduce
$$
\psi_{E,z}(x,y)=\psi(x,y)\big(1+O( \tfrac{\log (u+1)}{\log y})\big)
\Big(\sum_{
p\in Q,\,
p\leq y
}
1/p\Big)\cdot\Big(\sum_{p\leq y}1/p\Big)^{-1}\,.
$$
In order to apply Lemma~\ref{lemma:density}, we invoke Theorem~\ref{th:main}(2) which asserts that
$$
\frac{\psi_E(y,z)}{|\Pi(x)|}-\rho(v)\ll \delta(y)v\,.
$$
Here we fix $\beta>0$ such that $\delta(y)v\ll (\log_2 y)^{-1-\beta}\log_3y(\log_4y)^{-1}$. This is $\ll (\log_2 y)^{-1-\frac \beta 2}$. Therefore, by Lemma~\ref{lemma:density} we have that
$$
\Big(\sum_{
p\in Q,\,
p\leq y
}
1/p\Big)\cdot\Big(\sum_{p\leq y}1/p\Big)^{-1}=\rho(v)(1+o(1))+O\big((\log_2 y)^{-1-\frac\beta 2}\big).
$$
Hence we infer
\begin{align*}
\psi_{E,z}(x,y)&=x\rho(u)\Big(1+O( \frac{\log (u+1)}{\log y})\Big)\rho(v)(1+o(1))+O\big((\log_2 y)^{-1-\frac\beta 2}\big)\\
& =x\rho(u)\rho(v)(1+o(1))\,.
\end{align*}
\end{proof}

\section{The case of non CM elliptic curves}\label{sec:heuristic}

It is interesting to investigate potential analogues of Theorem~\ref{th:equivalent} in the non CM case. This section suggests such an analogue and highlights its theoretical limitations.
Let $E/\Q$ be a non CM elliptic curve. Deuring's Theorem (Lemma~\ref{lemma:CM}) enabled us in the CM case to relate $\psi_E(x,y)$ to the count of primes in arithmetic progressions. In the non CM case a natural choice for the analogous prime counting function is the following:
\[
\pi_E(x;d)=|\{ p\leq x\colon d\mid |E(\F_p)| \}|\,.
\]
In celebrated work~\cite{Serre1972}, Serre shows the existence of an integer $M_E$ depending only on $E$, such that for $n$ coprime with $M_E$, the Galois group $G_n$ of the $n$-torsion field extension $E[n](\overline{\Q})/\Q$ is isomorphic to the full group ${\rm GL}_2(\Z/\Z)$. Moreover one has additional multiplicativity property $G_{mn}\simeq G_m\times G_n$ for any $m$ coprime with $n$.
David and Wu~\cite[Proof of Lemma 4.1]{DavidWu2012b} give an asymptotic development under GRH for the Dedekind zeta function of $\Q(E[d](\overline{\Q}))$ when $d$ is coprime to $M_E$ and squarefree: 
\begin{align}\label{eq:DW}
\pi_E(x;d)&=\frac{w(d)}d\frac{x}{\log x}+O_E(d^{3/2}x^{1/2}\log(dx))\,,\\
w_E(d)&=\prod_{\substack{\ell\mid d\\ \ell\text{ prime}}}\frac{\ell^2(\ell^2-2)}{|{\rm GL}_2(\Z/\ell\Z)|}
=\prod_{\substack{\ell\mid d\\ \ell\text{ prime}}}\frac{\ell(\ell^2-2)}{(\ell-1)(\ell^2-1)}
\,.\nonumber
\end{align}
In the spirit of the Bombieri--Vinogradov Theorem and of its expected generalization Conjecture~\ref{conjecture:EH}, it is tempting to expect some strong average version of~\eqref{eq:DW} over $d$. The following results are evidence for the validity of this ``Elliott--Halberstam phenomenon'' that we next state (Hypothesis~\ref{hypothesis} below).

\begin{theorem} \label{th:BVnonCM}Let $E/\Q$ be a non CM elliptic curve and assume the GRH for Dedekind Zeta functions.
\begin{enumerate}
    \item One has (\cite[Prop. 3.8]{Kowalski2005}):
    $$
    \sum_{d\leq x^{\frac 14}/(\log x)^2}\varphi(d)|\{p\leq x\colon E[d](\overline{\Q})\subset E(\F_p)\}|=\Bigg(\sum_{d\geq 1}\frac{\varphi(d)}{|G_d|}\Bigg)\frac{x}{\log x}+O_E\Big(\frac x{(\log x)^3}\big)
    $$
\item One has (\cite[(4.7)]{DavidWu2012b}): 
$$
\sum_{\substack{d\leq x^{\frac 15}/(\log x)^4\\ p\mid d\Rightarrow  M_E<p\leq x^{1/10}/(\log x)^4}}2^{\omega(d)}\mu(d)^2\Big|\pi_E(x;d)-\frac{w_E(d)}d\frac{x}{\log x} \Big| \ll_{E}  \frac{x}{(\log x)^3} \,.
$$
\end{enumerate}
\end{theorem}
Regarding point 2 of Theorem~\ref{th:BVnonCM}, we follow Pollack who studied the elliptic curve analogue of the Titchmarsh divisor problem (\cite[p.185]{Pollack2016}): 
\begin{center}
\textit{
``We pretend that this approximation
is valid for $d$ up to size $\approx x$, at least on average''}.
\end{center}
This gives rise to the following hypothesis inspired by Conjecture~\ref{conjecture:EH}.

\begin{hypothesis}\label{hypothesis}
Let $E/\Q$ be a non CM elliptic curve. Then one has:
 \[
 \sum_{d\leq X}
\Big|\pi_E(x;d)-\frac{w_E(d)}d\frac{x}{\log x} \Big| \ll_{E,\omega}  \frac{x}{(\log x)^\omega}\,,
 \]
for any $X\leq x^{1-\delta}$, $x\geq 2$, for any constant $\omega> 0$, and where one extends $w_E$ to a function on $\N$ satisfying $w_E(mn)=w_E(m)w_E(n)$ for any coprime integers $m,n$ such that either $m$ or $n$ is coprime to $M_E$ (see~\cite[\S 2]{DavidWu2012b}).
\end{hypothesis}

Note that Hypothesis~\ref{hypothesis} allows any exponent $\omega>0$ on the denominator of the upper bound. This mimicks the upper bound appearing in the Elliott--Halberstam conjecture~\ref{conjecture:EH}; moreover we believe that there was no attempt to optimize the exponent $3$ appearing in the upper bounds of Theorem~\ref{th:BVnonCM} in the works of Kowalski and David--Wu. Finally, as in the proof of Theorem~\ref{th:equivalent}, we need an exponent $\omega>6$ to conclude the proof of Theorem~\ref{th:main nonCM}.

Hypothesis~\ref{hypothesis} enables us to prove the following analogue of Theorem~\ref{th:equivalent}.
\begin{theorem}\label{th:main nonCM}
For any $x\geq 2$ and $y\in[1,x]$ we set $u=\frac{\log x}{\log y}$. Assume Hypothesis~\ref{hypothesis} for a non CM elliptic curve $E/\Q$. 
Then, as $x\rightarrow \infty$ and $y$ is such that $u\leq u_0$, for some fixed constant $u_0$, we have 
\[
\psi_E(x,y) \sim \frac{x}{\log x}\rho(u)\,. 
\]

\end{theorem}

\begin{proof}
The argument is a verbatim translation of the proof of Theorem~\ref{th:equivalent}. We fix $\delta>0$.
By M\"obius inversion we split the studied prime counting function:
\[
\psi_E(x,y)=\Big|\Big\{p\leq x\colon {\rm gcd}\big(|E(\F_p)|,\prod_{\substack{\ell\text{ prime}\\ \ell>y}}\ell\big)=1\Big\}\Big|=S_1+S_2,
\]
where 
\[
S_1=\sum_{\substack{
q\leq x^{1-\delta}\\
P^-(q)>y}}
\mu(q)\pi_E(x;q) \,,\qquad
S_2=  \sum_{
\substack{x+2\sqrt{x}\geq q> x^{1-\delta}\\
P^-(q)>y}}
\mu(q)\pi_E(x;q)   \,.
\]
Note that the upper bound on $q$ in the index set of $S_2$ comes from the Hasse--Weil bound on $|E(\F_p)|$.
We next decompose $S_1=S_1'+S_1''$ where
\[
S_1'=\frac{x}{\log x}\sum_{
\substack{
q \leq x^{1-\delta}\\
P^-(q)>y}
}
 \frac{\mu(q)w_E(q)}{q}\,,\qquad
S_1''=\sum_{\substack{
q \leq x^{1-\delta}\\
P^-(q)>y}
} 
\mu(q)r(x,q)\,,
\]
and where $r(x,q)=\pi_E(x;q)-\frac{x}{\log x}\frac{w_E(q)}{q}$. Since $u$ remains bounded, we may assume that $M_E<y\leq x$, so that $w_E$ is multiplicative on all integers 
$q$ such that $P^-(q)>y$. If in addition $q$ is squarefree, the formula~\eqref{eq:DW} for $w_E(q)$ is valid and yields $w_E(q)=1+O((P^-(q))^{-1})$. We compute
\begin{align*}
    \frac{S_1'}{x(\log x)^{-1}}&=\sum_{
\substack{
q \leq x^{1-\delta}\\
P^-(q)>y\,,
\gcd(q,M_E)=1
}}\frac{\mu(q)w_E(q)}{q}
=\sum_{\substack{
q \leq x^{1-\delta}\\
P^-(q)>y}}\frac{\mu(q)}q+O\Big(\sum_{\substack{
q \leq x^{1-\delta}\\
P^-(q)>y}}\frac 1{qP^-(q)}\Big)\\
&\sim \rho(u)+O\Big(\frac 1y\sum_{q\leq x^{1-\delta}}\frac 1q\Big)\qquad (x\to\infty,\, u\ll 1)\,,
\end{align*}
where we have used Lemma~\ref{lemma:LT}. Finally, the fact that $u\ll 1$ implies that the error term is $O((\log y)/y)=o(1)$ as $x\to \infty$. This establishes $S_1'\sim \rho(u)\frac{x}{\log x}$ as $x\to\infty $ with $u\ll 1$. 


As in the proof of Theorem~\ref{th:equivalent} (step $1''$), we show that $S_1''=O(x/(\log x)^\omega)$ by virtue of Hypothesis~\ref{hypothesis}. In particular, for bounded $u$, one has $S_1''=o(\frac{x}{\log x}\rho(u))$.

We turn to the evaluation of $S_2$. As in the proof of Theorem~\ref{th:equivalent}, we use Lemma~\ref{lem:Rosser-Iwaniec} to obtain the following upper bound: $|S_2|\ll S_2'+S_2''$. Here we define
\[
S_2'=\sum_{
m\leq x^\delta}
 \frac{x}{\log x}\frac{w_E(m)}{m} \prod_{
\substack{
p\in\cP,\,p<z\\
p \nmid mM_E}
}  \left( 1-\frac{w_E(p)}{p} \right)\,,\qquad
S_2''=\sum_{
m\leq x^\delta
}  \sum_{\substack{
d<D \\
d\mid P(z)
} }  |r(x,md)|\,,
\]
where $z=D=1-2\delta$, the parameter $ca'$ has to be replaced by $M_E$ to define the set of primes $\cP$ and where, for $p\in\cP$, we choose $w(p)=w_E(p)$ (which satisfies the hypotheses of Lemma~\ref{lem:Rosser-Iwaniec} as shown in~\cite[Proof of Lemma 4.1]{DavidWu2012b}). We denote by $P(z)$ the product of primes in $\cP$ that are less than $z$. The fact that $(1-\frac{w_E(p)}{p})/(1-\frac{1}{p})=1+O(\frac{1}{p^2})$ for $p\in\cP$ implies that the method used in the proof of Theorem~\ref{th:equivalent} to handle the contribution of $S_2'$ also yields in the present case\footnote{Alternatively, one could appeal to~\cite[(4.3), (4.9)]{DavidWu2012b} to estimate the inner product over primes in the upper bound for $S_2'$.} that $S_2'=O(\delta \frac{x}{\log x}u)$.

Finally, to obtain the bound for $S_2''=o(\frac{x}{\log x})$ we argue as in the proof of Theorem~\ref{th:equivalent}, invoking Hypothesis~\ref{hypothesis} for a fixed $\omega>6$. We conclude by letting $\delta\to 0$.



\end{proof}

\appendix
\section{Alternative proof of Corollary~\ref{cor:main}}\label{appendix:main}



We prove the following form of Corollary~\ref{cor:main}, using the same method as for Theorem~\ref{th:equivalent} (\emph{i.e} an adaptation of Wang's approach~\cite{Wang2018}). 

\begin{proposition}\label{prop:thmain2}
Let $K$ be $\Q$ or an imaginary quadratic field of class number~$1$. Let $a,c\in\OO_K$ be fixed and let $\mu\in \OO_K^\times$. Let $C\in (0,\frac 12)$ and $\beta\in (2C,1)$. Assuming Conjecture~\ref{conjecture:parametricEH}, we have 
\[
\psi_K(x,y;c,a,\mu)=\frac{x}{\varphi(c)\log x}\rho(u)  \left(1+O\left( \delta(x)^{1-\beta} u\log u) \right)\right)\qquad (x\to \infty) 
\]
uniformly for $2\leq u:=\log x/\log y\leq C\log_2x/\log_3 x$ and $\delta(x)\in (\frac{\log_2 x}{\eta \log x},(\log x)^{-2C/\beta})$.
\end{proposition}

\begin{proof}


Let $\varepsilon(x,y)=\delta u \log u$.  Recall that $\delta(x)$ satisfies the assumptions of Conjecture~\ref{conjecture:parametricEH}.

We follow through the steps of the proof of Theorem~\ref{th:equivalent}. We split $\psi_K(x,y;c,a,\mu)=S_1+S_2$ with $S_1$ and $S_2$ as in Equation~\eqref{eq:S1 and S2} (up to replacing $\delta$ by $\delta(x)$). We write $S_1=S_1'+S_1''$ as in Equation~\eqref{eq:S1' and S''} and we upper bound $|S_2|\leq S_2'+S_2''$ as in Equation~\eqref{eq:S2' and S2''}.

{\bf Step $1'$.} We show that $S_1'=\frac{x}{\varphi(c)\log x}\rho(u)(1+ o(\varepsilon(x,y)))$. By Lemma \ref{lemma:LT} we have   
\begin{equation}\label{eq:asympS1p}
S_1'= \frac{x}{\varphi(c)\log x} \sum_{ \substack{ q\in \OO_K,\, \norm{q} \leq x^{1-\delta}\\P^-(\norm{q})>y } } \frac{\mu(q)}{\varphi(q)}
=\frac{x}{\varphi(c)\log x}\Big( \rho\big(\frac{\log(x^{1-\delta(x)})}{\log y}\big) + O(\exp(-(\log y)^{\frac35-\epsilon}))\Big)\,.
\end{equation}
Note that $u\ll \log_2x/\log_3 x$ implies that $\log y\gg \log x\log_3 x/\log_2 x$ (so that $y$ lies in the range of validity for Lemma~\ref{lemma:LT}) 
In particular the error term in~\eqref{eq:asympS1p} is $o( \delta(x)\rho(u)u\log u)$. Indeed one has
\begin{align*}
\frac{{\rm e}^{-(\log y)^{\frac35-\epsilon}}}{\rho(u)}&\ll
{\rm e}^{-(\log y)^{\frac35-\epsilon}+2u\log u}\\
& \ll\exp\Big(-\big(\log x\frac{\log_3 x}{\log_2 x}\big)^{\frac 35-\epsilon}+
O\big(\log_2 x(1-\frac{\log_4 x}{\log_3x})\big)\Big)
\ll 1
\end{align*}
where we have also used the lower bound $\rho(u)\gg \exp(-2u\log u)$ coming from~\eqref{eq:HT}. 

Finally, using~\eqref{eq:HT} again, we compute (writing $\delta$ instead of $\delta(x)$, for simplicity):
\[
\log\Big( \rho\big(\frac{\log(x^{1-\delta})}{\log y}\big)\Big)-\log (\rho(u))  ) = -u(1-\delta)(\log u+\log (1-\delta)) + u\log u +O(u\log_2 u)\,.
\]
Hence $\rho\left(\frac{\log(x^{1-\delta(x)})}{\log y}\right)/\rho(u)=1+O(\delta(x) u\log u)$ which proves the stated estimate for $S_1'$.

{\bf Step $1''$.} We show that $S_1''=o(\frac{x}{\log x}\rho(u)\varepsilon(x,y))$. To do so
we apply Conjecture~\ref{conjecture:parametricEH} in the same way we applied Conjecture~\ref{conjecture:EH} in the proof of Theorem~\ref{th:equivalent}. We fix $\omega>2C-2$ where $C$ is an absolute constant such that we work under the restriction $u\leq C\log_2x/\log_3 x$. The exact same argument as the one used to obtain~\eqref{eq:S1''} yields $S_1''=O(x(\log x)^{-1}/(\log x)^{\omega-1})$. The point is that Conjecture~\ref{conjecture:parametricEH}
 asserts that the implied constant in this upper bound is uniform in $u$. We now compute, using again the bound $\rho(u)\gg \exp(-2u\log u)$ and the fact that $u\log u$ grows as $x\to \infty$,
 \begin{align*}
 \frac{S_1''}{x(\log x)^{-1}\rho(u)\varepsilon(x,y)}&\ll \frac {{\rm e}^{2u\log u}}{(\log x)^{\omega-1}\delta(x)u\log u}\ll \frac{(\log _2 x){\rm e}^{2u\log u}}{(\log x)^{\omega-2}u\log u}\\
 &\ll \frac{(\log _2 x)\exp(2C\log_2 x)}{(\log x)^{\omega-2}}\ll \frac{\log_2 x}{(\log x)^{\omega-2C-2}}  \,.
 \end{align*}


\medskip
\medskip

{\bf Step $2'$.} We prove that $S_2'=O\big(\frac{x}{\log x}\frac{\delta(x)}{1-2\delta(x)}u\big)$ in the exact same way as for~\eqref{eq:S2'} (where the implied constant is absolute). From the bound $1-2\delta(x)\geq 1-2\eta\gg 1$ (recall Conjecture~\eqref{conjecture:parametricEH}),
we conclude that
$$
\frac{\delta(x)^\beta S_2'}{\rho(u)x(\log x)^{-1}\varepsilon(x,y)}\ll \frac{\delta(x)^\beta}{\rho(u)}\ll
\delta(x)^\beta{\rm e}^{2C\log_2x}\ll 1\,. 
$$
Therefore we have $S_2'=O(\frac{x}{\log x}\delta(x)^{1-\beta}\rho(u)u\log u)$.


\medskip
\medskip

{\bf Step $2''$.} We prove that $S_2''=o(\frac{x}{\log x}\rho(u)\varepsilon(x,y))$.
As in the proof of~\eqref{eq:S2''}, we have 
$$
S_2''=O\Big(\frac{x(\log x)^{-1}}{(\log x)^{\frac \omega 2 -3}}\Big)
$$ 
with an implied constant depending only on $\omega$. We then argue as in Step $1''$ above by requiring this time that $\omega>4C+4$. This concludes step $2''$ and the proof of Proposition~\ref{prop:thmain2}.
\end{proof}






\section{Numerical illustration}\label{sec:numerical}


We consider the examples $E_7\colon y^2 + x y = x^{3} - x^{2} - 2 x - 1$ and $E_{11}\colon y^2 + y = x^{3} - x^{2} - 7 x + 10 $ which have endomorphism rings included in $\Q(\sqrt{-7})$ and $\Q(\sqrt{-11})$, respectively.

Our numerical experiment (see Figure~$1$) can be seen as a type of \emph{Chebyshev race}\footnote{In an 1853 letter, Chebyshev observes that the count of primes up to $x$ that are $3$ modulo $4$, almost always exceeds that of primes that are $1$ modulo $4$. Modern instances of what is now called a ``prime number race'' have been extensively studied in the recent years.} between $E_7$ and $E_{11}$: we compare $\psi_{E_7}(x,2^7)$ and $\psi_{E_{11}}(x,2^7)$ for various values of $x$. The data shows that $E_7$ is ``always ahead'', in other words, $E_7$ is  more ECM-friendly than $E_{11}$ for these values of $x$ and $y$. We repeat this Chebyshev race for $y=2^{25}$ and obtain the same conclusion. This suggests the following conjecture: $E_{7}$ is more ECM-friendly than $E_{11}$ uniformly for $x$ and $y$ when $y$ grows with $x$ and is not too large compared to~$x$.

In Figure~$2$ we search for an accurate expression for the error term in the asymptotic expansion of $ \psi_E(x,y)/\psi_E(x,\infty)$. First we plot the expression which is given by Scourfield's Theorem~\ref{th:Scourfield} $$\psi_K(x,y)/\psi_K(x,\infty)=\rho(u)\left(1-\frac{\log(u+1)}{\log y}\gamma_K(1+o(1))\right).$$ The data corroborates the accuracy of this theoretical value. We also plot the expression $\widetilde{\gamma}_K:=\frac{\frac{\psi_K(x,y)}{\psi_K(x,\infty)}\rho(u)^{-1}-1}{\log(u+1)/\log y}$ which converges to a constant when $u=1.5$.  
We define similarly $\widetilde{\gamma}_E:=\frac{\frac{\psi_E(x,y)}{\psi_E(x,\infty)}\rho(u)^{-1}-1}{\log(u+1)/\log y}$ for which we don't have a theoretical result. The data suggests that $\log (\psi_E(x,y)/\psi_E(x,\infty))$ has the same main error term as $\log( \psi(x,y)/\psi(x,\infty))$ and that $\widetilde{\gamma}_E$ converges as $x\rightarrow \infty$. Note that if  $\widetilde{\gamma}_E$ does converge, it is not expected to be equal to $\widetilde{\gamma}_K$ since the former quantity involves the contribution of both split and inert primes. 

We emphasize that the collection of data presented in Figure~$1$ requires the computation of $|E(\F_p)|$ for a large number of primes $p$ and for $E=E_7$ or $E_{11}$. We haven't used the expensive Schoof algorithm because, in the particular case of CM curves, specific methods exist. We haven't used the formul\ae\ in~\cite[Th 5.3, Th 5.5 and Th 5.6]{RubinSilverberg2009} either. Indeed the characters involved, although explicit, have a costly evaluation, moreover the formul\ae\ are prone to typos. Instead we follow a classical procedure and use Lemma~\ref{lemma:CM}: let $p$ be a prime for which one wishes to compute $|E(\F_p)|$. Fix a Weierstrass model for $E$ and pick a dozen random projective points $P_1\in E(\F_p)$; for each root of unity $\mu$ of $K$ (there are at most $6$ of them), compute the possible value of $N:=\norm{\pi-\mu}$. Next compute $|N]P_1$, $[N]P_2$, $\ldots$ and rule out $N$ if one of the computed points is not $(0:1:0)$. When all but one of the possible values of $|E(\F_p)|$ have been ruled out, one successfully outputs the result.   

The data plotted in Figures~$1$ and $2$ is available online at:
\begin{center}
\begin{small}
\url{https://razvanbarbulescu.pages.math.cnrs.fr/ElliottHalberstam/ElliottHalberstam.html}
\end{small}
\end{center}

\includepdf[pages={1-},scale=1.1,pagecommand={}]{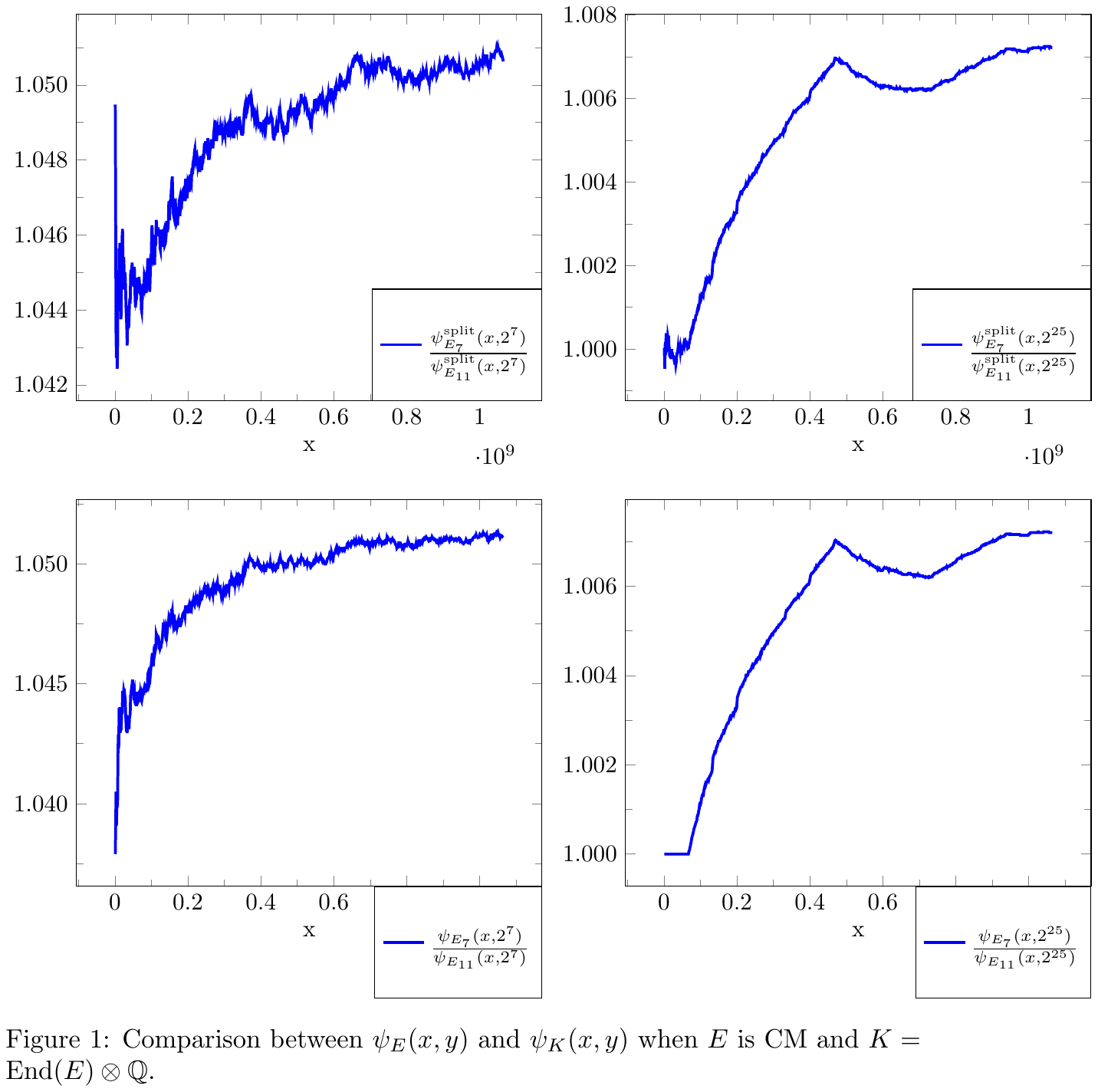}

\bibliographystyle{alpha}
\bibliography{references}

\end{document}